\documentclass[12pt]{article}
\usepackage{graphicx,caption,setspace,nicefrac}
\usepackage{hyperref,amsthm,epigraph,fancyhdr}
\usepackage[margin=1.25in]{geometry}
\usepackage[autostyle]{csquotes}
\usepackage[style=authoryear,backend=bibtex]{biblatex}

\usepackage{heuristica}
\usepackage[heuristica,vvarbb,bigdelims]{newtxmath}

\hypersetup{colorlinks=true, citecolor=blue,urlcolor=blue}
\MakeOuterQuote{"}
\newtheorem{definition}{Definition}
\newtheorem{theorem}{Theorem}
\newtheorem{lemma}{Lemma}
\newtheorem{example}{Example}
\newtheorem{corollary}{Corollary}
\newtheorem{proposition}{Proposition}
\newtheorem{remark}{Remark}

\DeclareMathOperator*{\argmax}{arg\,max}

\begin{document}
\begin{singlespacing}
\title{{\textbf{Grade Inflation and Stunted Effort in a Curved Economics Course}}}
\author{\sc Dr. Alex Garivaltis\footnote{Assistant Professor, Department of Economics, School of Public and Global Affairs, College of Liberal Arts \& Sciences, Northern Illinois University, 514 Zulauf Hall, DeKalb, IL 60115.  E-Mail:  \href{mailto:agarivaltis1@niu.edu}{\tt agarivaltis1@niu.edu}.  Homepage:  \url{http://garivaltis.com}.  ORCID-iD:  0000-0003-0944-8517.}}
\maketitle
\abstract{To protect his teaching evaluations, an economics professor uses the following exam curve:  if the class average falls below a known target, $m$, then all students will receive an equal number of free points so as to bring the mean up to $m$.  If the average is above $m$ then there is no curve; curved grades above $100\%$ will never be truncated to $100\%$ in the gradebook.  The $n$ students in the course all have Cobb-Douglas preferences over the grade-leisure plane;  effort corresponds exactly to earned (uncurved) grades in a $1:1$ fashion.  The elasticity of each student's utility with respect to his grade is his ability parameter, or relative preference for a high score.  I find, classify, and give complete formulas for all the pure Nash equilibria of my own game, which my students have been playing for some eight semesters.  The game is supermodular, featuring strategic complementarities, negative spillovers, and nonsmooth payoffs that generate non-convexities in the reaction correspondence.  The $n+2$ types of equilibria are totally ordered with respect to effort and Pareto preference, and the lowest $n+1$ of these types are totally ordered in grade-leisure space.  In addition to the no-curve (``try-hard'') and curved interior equilibria, we have the ``$k$-don't care'' equilibria, whereby the $k$ lowest-ability students are no-shows.  As the class size becomes infinite in the curved interior equilibrium, all students increase their leisure time by a fixed percentage, i.e., $14\%$, in response to the disincentive, which amplifies any pre-existing ability differences.  All students' grades inflate by this same (endogenous) factor, say, 1.14 times what they would have been under the correct standard.   
}
\end{singlespacing}
\titlepage

\epigraph{\onehalfspacing ``Just as eating contrary to the inclination is injurious to the health, study without desire spoils the memory, and it retains nothing that it takes in.''}{\textemdash \onehalfspacing Leonardo da Vinci}

\epigraph{\onehalfspacing ``I think the big mistake in schools is trying to teach children anything, and by using fear as the basic motivation. Fear of getting failing grades, fear of not staying with your class, etc. Interest can produce learning on a scale compared to fear as a nuclear explosion to a firecracker.''}{\textemdash \onehalfspacing Stanley Kubrick}

\newpage

\section{A Game-Theoretic Model of the University.}

We assume that there are $n$ students in the course, called $i\in\left\{1,2,...,n\right\}$, where $n\in\mathbb{N}$ and $\mathcal{I}:=\left\{1,2,...,n\right\}$ is the set of players for our $n$-person game.  Each student $i$'s commodity space\footnote{Due to the presence of an exam curve, it will be possible to earn a grade that is higher than $100\%$;  this is why we write the commodity space as $\mathbb{R}_+\times[0,1]$ instead of $[0,1]\times[0,1]$.  As we will see below, the highest possible curved grade that can ever occur in the model is $\left(200-100/n\right)\%$.} $\mathbb{R}_+\times[0,1]$ will consist of two goods:  the exam grade, $G_i$, expressed as a percentage, and leisure, laziness, or non-effort, expressed as a percentage $L_i\in[0,1]$.  Each student chooses an effort level,  $x_i\in[0,1]$, where we have the resource constraint $L_i=1-x_i$.  In the absence of a curve or other type of grade-bloating scheme imposed by the professor, effort will be assumed to correspond perfectly with exam performance;  say, if you give $x_i=80\%$, then your exam grade will be $G_i=80\%$ and your leisure allocation will be $L_i=20\%$.  Each student $i$ will  be assumed to have Cobb-Douglas preferences $\succsim_i$ over the unit square (\cite{cobbdouglas});  the parameter $\alpha_i\in(0,1)$ will denote the elasticity of student $i$'s utility with respect to his or her grade:
\begin{equation}\label{normalized}
U_i(G_i,L_i):=G_i^{\alpha_i}L_i^{1-\alpha_i},
\end{equation}so that $1-\alpha_i$ is the elasticity of his utility with respect to leisure, or non-effort.  Thus, each student's preferences $\succsim_i$  are represented via a linearly homogeneous utility function, whereby we have made the normalization $U_i(1,1)=1$.  Thus, a $100\%$ utility index corresponds to receiving a grade of $100\%$ whilst exerting no effort.

\begin{figure}[t]
\begin{center}
\includegraphics[height=200px]{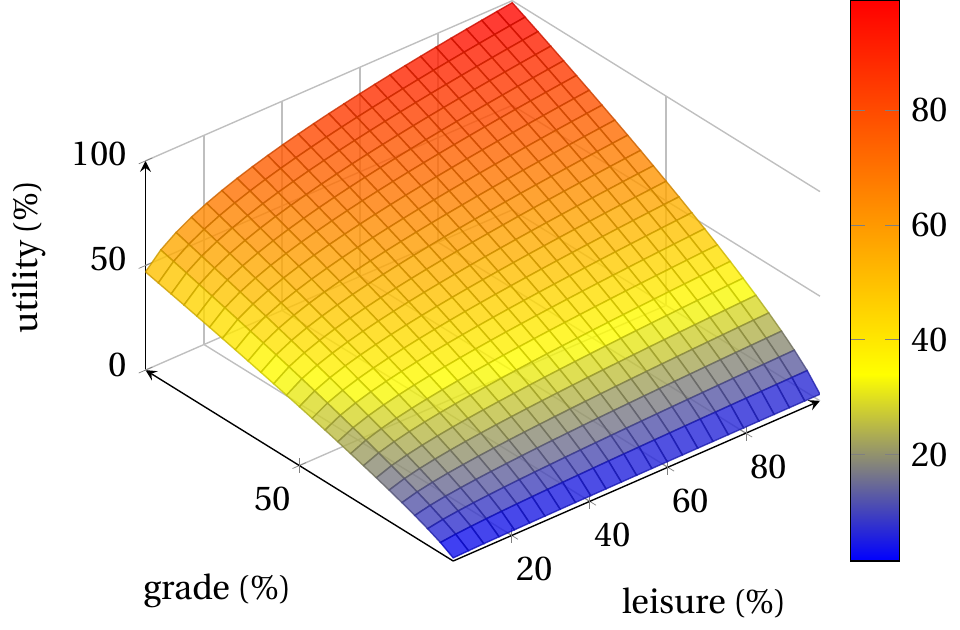}
\caption{\sc Student preferences over a commodity space of $\left(\text{grade},\text{leisure}\right)$ pairs.  This particular student (a C student) has the utility $U_i(G_i,L_i)=G_i^{0.75}L_i^{0.25}$.}
\end{center}
\end{figure}

Under ordinary (uncurved) conditions, student $i$ will choose a level $x_i^*$ of study intensity that solves
\begin{equation}\label{effort}
\max_{0\leq x_i\leq1}x_i^{\alpha_i}\left(1-x_i\right)^{1-\alpha_i},
\end{equation}or, equivalently, that solves
\begin{equation}\label{logs}
\max\limits_{0\leq x_i\leq1}\left(\alpha_i\log(x_i)+(1-\alpha_i)\log\left(1-x_i\right)\right).
\end{equation}The optimization problem (\ref{effort}) must have an interior solution, since the endpoints give zero utility.  Differentiating (\ref{logs}), we have the unique optimum $x_i^*=\alpha_i$, viz., each student's preference parameter $\alpha_i$ is precisely the grade that he or she would receive without the distortion of an exam curve.  Say, a student with the utility function $U(G_i,L_i)=G_i^{0.85}L_i^{0.15}$ would wind up with $85\%$ on the exam, for a solid B.  Thus, under ordinary conditions, the mean on the exam will be $\overline{x}:=(1/n)\sum\limits_{i=1}^nx_i=(1/n)\sum\limits_{i=1}^n\alpha_i=\overline{\alpha}$, viz., the sample mean of the students' preference parameters.

\begin{figure}[t]
\begin{center}
\includegraphics[height=200px]{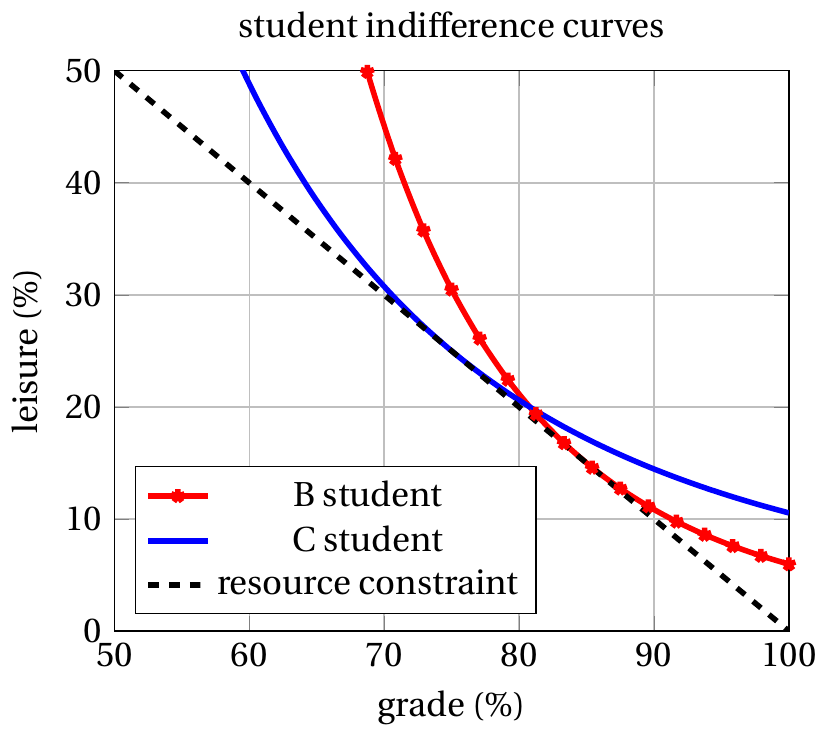}
\caption{\sc Indifference curves of B ($\alpha_i:=85\%$) and C ($\alpha_i:=75\%$) students in the grade-leisure plane.  Without an exam curve, we have the resource constraint $G_i+L_i=1$.}
\end{center}
\end{figure}

\begin{figure}[t]
\begin{center}
\includegraphics[height=200px]{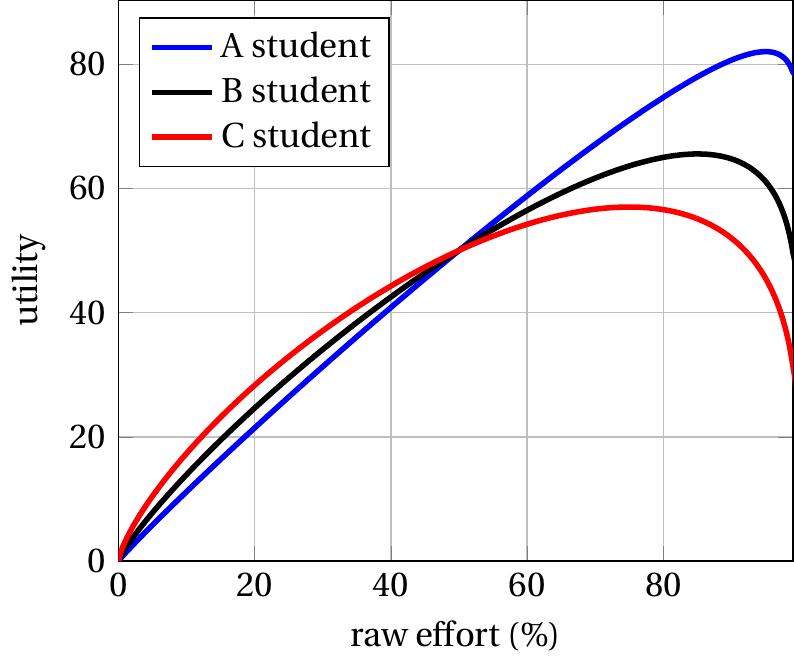}
\caption{\sc The utility of effort, $U_i(x_i)=x_i^{\alpha_i}\left(1-x_i\right)^{1-\alpha_i}$, for three different students, absent the distortions of an exam curve (i.e., $m:=0$).  The students' respective Cobb-Douglas parameters are $\left(\alpha_1,\alpha_2,\alpha_3\right):=\left(95\%,85\%,75\%\right)$.}
\end{center}
\end{figure}

\par
The Professor of the course will be assumed to use the following curving system, known to all the students in advance:  all students will receive an identical number of additional exam points in order to raise the mean $\overline{x}$ to an acceptable level, $m$.  If $\overline{x}\geq m$, then no points will be added;  in the event of a curve, any students whose curved scores exceed $100\%$ will not have their grades truncated\footnote{The author's reasoning for this feature is as follows:  under the present curving system, all students receive an equal number of free points $m-\overline{x}$, and the best students in the class are not being penalized for their high effort levels.  Naturally, the curve $\left(x_i,x_{-i}\right)\mapsto x_i+\max\left(m-\overline{x},0\right)$ already provides an incentive to decrease one's effort;  the exact magnitude of each student's decrease in study intensity will be determined below in equilibrium.  Allowing scores that are higher than $100\%$ net of the curve will turn out to be a useful device for educational inferences, i.e., we can calibrate the model by assuming that the best student in the class had an ability parameter of $\alpha_i=1$. } to $100\%$.  That is, a curved score of $105\%$ will figure into a given student's utility function as $1.05^{\alpha_i}(1-x_i)^{1-\alpha_i}$ utils.  Thus, player $i$'s curved grade $G_i=G_i\left(x_i,x_{-i}\right)$ consists in the expression  
\begin{equation}\boxed{
G_i\left(x_i,x_{-i}\right)=x_i+\max\left(m-\overline{x},0\right)=\max\left(m+\frac{n-1}{n}\left(x_i-\overline{x}_{-i}\right),x_i\right).}
\end{equation}Which is to say, in the event that there is an exam curve, then student $i$'s grade will be given by 
\begin{equation}
G_i(x_i,x_{-i})=x_i+m-\frac{S_{-i}+x_i}{n}=x_i+m-\frac{(n-1)\overline{x}_{-i}+x_i}{n}=m+\frac{n-1}{n}\left(x_i-\overline{x}_{-i}\right),
\end{equation}where $S_{-i}:=\sum\limits_{j\neq i}x_j$ is the aggregate effort of the non-$i$ students, and $\overline{x}_{-i}:=S_{-i}/(n-1)$ is the average effort of the non-$i$ students.  Thus, for the game outcome $x:=\left(x_i,x_{-i}\right)$, player $i$'s utility (or payoff) amounts to
\begin{equation}\boxed{
U_i\left(x_i,x_{-i}\right)=\max\left(m+\frac{n-1}{n}\left(x_i-\overline{x}_{-i}\right),x_i\right)^{\alpha_i}\left(1-x_i\right)^{1-\alpha_i},}
\end{equation}where $x_{-i}:=\left(x_j\right)_{j\neq i}$ is the action profile of the student $i$'s opponents.  Hence, each player's action set $\mathcal{A}_i$ is the unit interval $[0,1]$, and the set of all game outcomes, or action profiles $x=(x_i)_{i=1}^n$, is the unit hypercube $\bigtimes\limits_{i=1}^n\mathcal{A}_i=[0,1]^n$.  This completes our formal definition of the game $\Gamma_n:=\left(\left\{1,2,...,n\right\},\left(\mathcal{A}_i\right)_{i=1}^n,\left(\succsim_i\right)_{i=1}^n\right)$ that is presently at hand (cf. with \cite{vnm,osborne}).  In terms of the aggregate resources that are up for grabs in our environment, we have $\sum\limits_{i=1}^nG_i\ge nm$ and $\sum\limits_{i=1}^n(G_i+L_i)\ge n$.  The parameter space $\Theta$ for the game-theoretic model $\Gamma_n:=\Gamma_n\left(\left(\alpha_i\right)_{i=1}^n,m\right)$ consists in the open set $\Theta:=(0,1)^{n+1}$, which has $n+1$ degrees of freedom.
\par
Thus, in our mathematical formalism for the university interactions that happen in actual life, the pupils' economic behavior hinges on the precise structure of the continuous mapping $x\mapsto\left(U_1(x),...,U_n(x)\right)$ that transforms $[0,1]^n$ into the positive orthant $\mathbb{R}_+^n$.  Our object of study, $\Gamma_n$, is an infinite, continuous game with compact and convex (uni-dimensional) strategy sets.  Although each player $i$'s payoff function $U_i(x)$ is continuous, it is not quasi-concave in his own strategy $x_i$;  thus, there is no guarantee that the reaction correspondence $\argmax\limits_{x_i\in[0,1]}U_i(x_i,x_{-i})$ is convex-valued.  This non-convexity (cf. with \cite{fudenbergtirole}) means that the usual method of proving the existence of pure strategy equilibria (via the \cite{kakutani} fixed point theorem, cf. with \cite{nash,nash51,debreu}) does not apply to our particular situation.  In general, the continuity and compactness imply the existence of a mixed-strategy equilibrium (cf. with \cite{glicksberg});  in order to make more specific conclusions, we will need to unravel the exact, concrete properties of our particular payoff $U_i\left(x_i,x_{-i}\right)$.
\par
In our world, the students' effort levels are strategic complements in the sense of \cite{bulow}:  if my classmates invest and exert themselves, and try very hard for a good score on the exam, then the curve will be less generous (or non-existent), which lowers my grade and increases the marginal utility of my effort.  Thus, my best reply is to join in with a higher effort level of my own.  This type of coordination also works in the reverse direction:  if my classmates (opponents) dump and tank their scores on the exam, then the padding of the curve allows me to decrease my effort level by some correct amount, thereby optimizing my welfare and reaping the benefits of increased leisure time.
\par
Away from the kink (viz., where $\overline{x}=m$), this strategic complementarity is manifest in the non-negativity of the cross partials $\partial^2\log\left(U_i\left(x\right)\right)/\left(\partial U_i\partial U_j\right)$ for $i\neq j$, i.e.,

$$\frac{\partial^2}{\partial x_i\partial x_j}\log\left(U_i\left(x\right)\right)=\alpha_i\frac{n-1}{n^2}\times\begin{cases}
0 & \text{if }\overline{x}>m\\
\text{undefined} & \text{if }\overline{x}=m\\
G_i^{-2} & \text{if }\overline{x}<m.\\
\end{cases}$$

However, since student $i$'s payoff $U_i(x)$ is not differentiable over the entirety of the hypercube $[0,1]^n$, we will avail ourselves of the more general and direct approach that is furnished by the Topkis theory (\cite{topkisbook,topkispaper,vives,milgrom}) of ordered comparative statics in supermodular games, as follows.
\begin{theorem}[Increasing Differences]
Each student's log-payoff has increasing differences, meaning that the utility that is gained from extra effort
\begin{equation}\label{differences}
\Delta\log U_i\left(x\right):=\log\left(U_i\left(x_i+\Delta x_i,x_{-i}\right)\right)-\log\left(U_i\left(x_i,x_{-i}\right)\right),
\end{equation}where $\Delta x_i>0$, is increasing in the efforts of all the non-$i$ players. 
\end{theorem}
\begin{proof}
Since each player's log-payoff depends only on the sample mean $\overline{x}_{-i}$ of his opponents' effort levels, it suffices to prove that the utility change (\ref{differences}) is increasing in $\overline{x}_{-i}$, since $\overline{x}_{-i}$ is increasing in $x_j$ for $j\neq i$.  Now, the leisure term $(1-\alpha_i)\log\left(1-x_i\right)$ is unaffected by the opponents' sample mean $\overline{x}_{-i}$;  accordingly, we define the relevant difference
\begin{multline}\label{f}
f\left(\overline{x}_{-i}\right):=\log\left(\max\left(m+\frac{n-1}{n}\left(x_i+\Delta x_i-\overline{x}_{-i}\right),x_i+\Delta x_i\right)\right)\\
-\log\left(\max\left(m+\frac{n-1}{n}\left(x_i-\overline{x}_{-i}\right),x_i\right)\right),
\end{multline}and proceed to show that $f(\bullet)$ is increasing over $[0,1]$.
\par
Now, the (univariate) function (\ref{f}) has a pair of distinct kinks, or points of non-differentiability, namely, $\left(nm-x_i-\Delta x_i\right)/(n-1)$ and $(nm-x_i)/(n-1)$.  These two points partition $[0,1]$ into three distinct intervals;  it suffices to show that $f(\bullet)$ is increasing over each such interval.  Over the interval $\overline{x}_{-i}\in\left[\left(nm-x_i\right)/\left(n-1\right),1\right]$, we have the constant function $f\left(\overline{x}_{-i}\right)\equiv\log\left(x_i+\Delta x_i\right)-\log(x_i)$ which is non-decreasing in $\overline{x}_{-i}$.  Next, in the event that $\overline{x}_{-i}\in\left[\left(nm-x_i-\Delta x_i\right)/\left(n-1\right),\left(nm-x_i\right)/\left(n-1\right)\right]$, we will have $f\left(\overline{x}_{-i}\right)=\log\left(x_i+\Delta x_i\right)-\log\left(m+\frac{n-1}{n}\left(x_i-\overline{x}_{-i}\right)\right)$, which is increasing in $\overline{x}_{-i}$ because student $i$'s grade is decreasing in $\overline{x}_{-i}$.  Finally, in the event that $\overline{x}_{-i}\leq\left(nm-x_i-\Delta x_i\right)/(n-1)$, our job amounts to demonstrating that the ratio
\begin{equation}
\frac{nm+(n-1)\left(x_i+\Delta x_i-\overline{x}_{-i}\right)}{nm+(n-1)\left(x_i-\overline{x}_{-i}\right)}=1+\frac{(n-1)\Delta x_i}{nm+n-1\left(x_i-\overline{x}_{-i}\right)}
\end{equation}is increasing in $\overline{x}_{-i}$, which is clearly true, since $\Delta x_i>0$ by hypothesis.  This completes the proof.
\end{proof}

\begin{figure}[t]
\begin{center}
\includegraphics[height=200px]{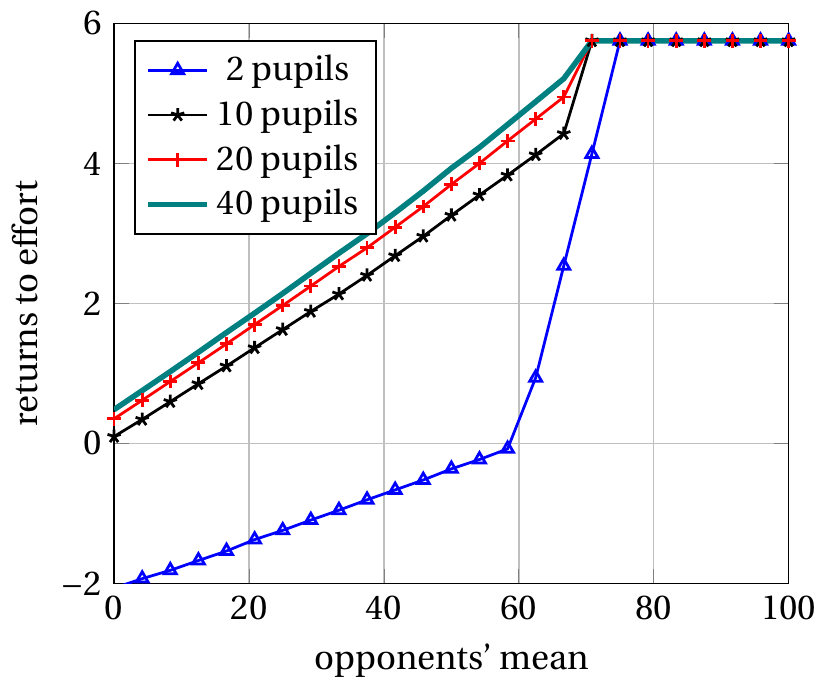}
\caption{\sc Increasing returns $\overline{x}_{-i}\mapsto\Delta U_i\left(x_i,\overline{x}_{-i}\right)=U_i\left(x_i+\Delta x_i,\overline{x}_{-i}\right)-U_i\left(x_i,\overline{x}_{-i}\right)$ to extra effort $\Delta x_i$ (``increasing differences'') for player $i$.  If $i$'s classmates study harder, then their sample mean $\overline{x}_{-i}$ increases, and pupil $i$'s benefit $\Delta U_i$ from the extra effort $\Delta x_i$ also increases.  Hence, we have strategic complements (and a supermodular game) in the author's classroom.  This illustration uses the parameters $\left(x_i,\Delta x_i,m,\alpha_i\right):=\left(65\%,15\%,70\%,85\%\right)$ and $n\in\{2,10,20,40\}$.  The middle piece of this (tripartite) function is an artifact of small class sizes;  it disappears in the limit as $n\to\infty$.}
\end{center}
\end{figure}

\begin{corollary}[Supermodularity of the $n$-Person Game $\Gamma_n$]
The game $\Gamma_n$ (that the author's students were playing for eight semesters) is supermodular with negative spillovers.  Hence, there exists at least one equilibrium point $x^e=\left(x^e_1,...,x^e_n\right)$ in pure strategies.  There exists a low-effort equilibrium $x^*$ and a high-effort equilibrium $y^*$ that bracket all Nash equilibria with respect to the vector partial order $\leq$ over $\mathbb{R}^n$.  That is, for every equilibrium point $x^e$, we have $x^*_i\leq x^e_i\leq y^*_i$ for all $i=1,...,n$.
\end{corollary}
\begin{proof}
Each player's payoff $U_i\left(x_i,x_{-i}\right)$ is continuous, it has increasing differences, and it is supermodular in his own action $x_i$, since every function of a single real variable is supermodular (cf. with \cite{berkeley,yildiz}).  Thus, the game $\Gamma_n$ is supermodular in the sense of \cite{topkispaper,topkisbook}.  Accordingly, the set of   fixed points of the best response correspondence is non-empty, and it has a greatest and least element with respect to the usual (coordinate-wise) partial ordering of $n$-dimensional Euclidean space (cf. with \cite{fudenbergtirole,yildiz}).  Since each player's payoff $U_i$ is a decreasing function of the opponents' sample mean $\overline{x}_{-i}$, it is therefore decreasing in the opposing action profile $x_{-i}$, so that we have negative spillovers (cf. with \cite{milgrom,levin}).
\end{proof}

\begin{figure}[t]
\begin{center}
\includegraphics[height=200px]{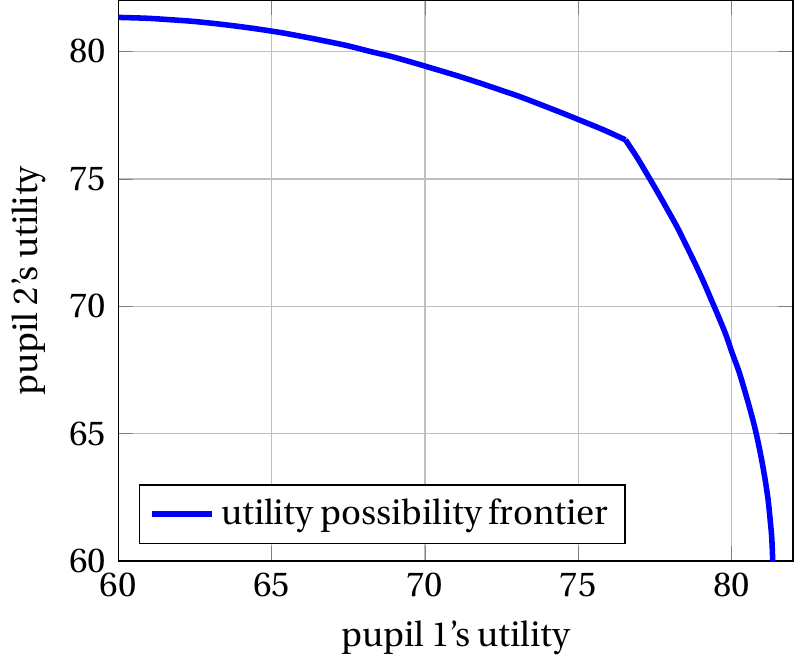}
\caption{\sc The utility possibility frontier for $n:=2$ students, under the parameter vector $\theta:=\left(\alpha_1,\alpha_2,m\right)=\left(75\%,75\%,70\%\right)$.  The set of Pareto efficient allocations $\left(x_1,x_2\right)$ is $\left([0,0.4]\times\{0\}\right)\cup\left(\{0\}\times[0,0.4]\right)$.}
\end{center}
\end{figure}

\begin{figure}[t]
\begin{center}
\includegraphics[height=200px]{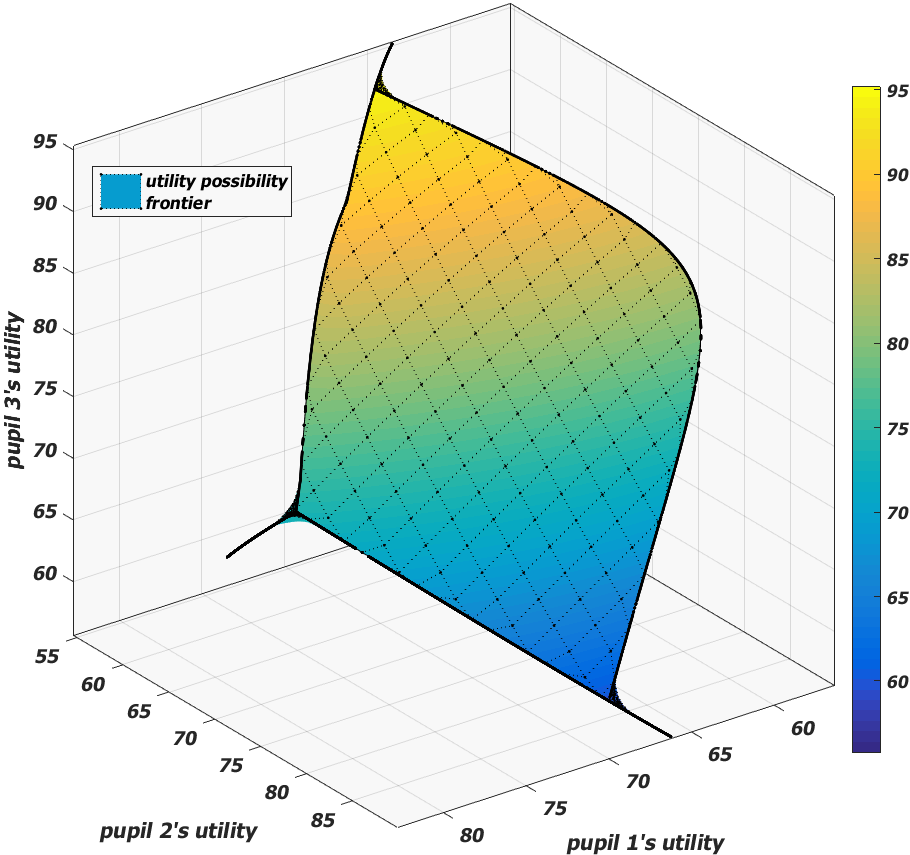}
\caption{\sc The 3-person utility possibility frontier (of undominated triples $\left(U_1,U_2,U_3\right)$ for the parameters $(\alpha_1,\alpha_2,\alpha_3,m):=(60\%,80\%,85\%,70\%)$. }
\end{center}
\end{figure}

\begin{figure}[t]
\begin{center}
\includegraphics[height=200px]{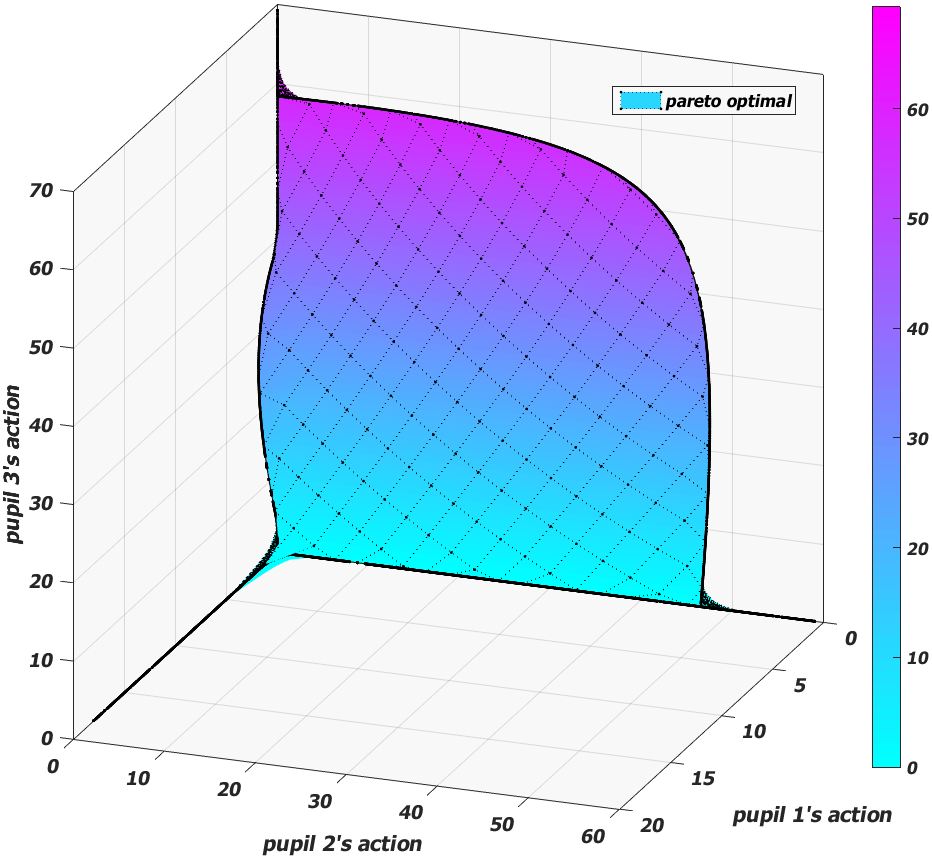}
\caption{\sc The set of pareto efficient allocations (``contract surface'') for the 3-person game, under the parameters $(\alpha_1,\alpha_2,\alpha_3,m):=(60\%,80\%,85\%,70\%)$.  These are very low-effort outcomes $(x_1,x_2,x_3)$ that will not generally obtain in equilibrium, due to the non-cooperative nature of the game.  If the students can all conspire to lose $\varepsilon$ points on the exam, then the curve will be increased by $\varepsilon$, washing away the loss for everybody.  Accordingly, in any pareto optimal allocation, one or more students must put zero effort.}
\end{center}
\end{figure}

In addition to the strategic complementarity, the pupils' returns to effort also respond monotonically to changes in the parameter values.  Say, when student $i$'s ability parameter increases, his relative preference for leisure decreases, and the returns to his effort increase;  similarly, when the professor's target mean decreases, the relative scarcity of exam points increases the extra utility that is gained from any fixed amount of additional effort, $\Delta x_i$.  Accordingly, we have the following Definition, which imposes the appropriate order structure on the parameter set, $\Theta$.
\begin{definition}[Lattice Structure of the Parameter Space]
Let $\theta:=(\alpha,m_1)\in\Theta$ and $\eta:=(\beta,m_2)\in\Theta$ denote two parameter vectors, where $\alpha,\beta\in(0,1)^{n}$ are the respective ability vectors and $m_1,m_2\in(0,1)$ are the respective target means.  We say that $\theta$ is harder than $\eta$, and we write $\theta\geq\eta$, if and only if $\alpha\geq\beta$ and $m_1\leq m_2$.  That is, the game is considered to get harder if any student's ability parameter increases or if the instructor's target mean decreases.  
\par
The partial ordering $\geq$, so defined, turns the parameter space into a lattice $\left(\Theta,\le,\vee,\wedge\right)$, whose join operation\footnote{The least upper bound $\alpha\vee\beta$ is the coordinate-wise maximum $\left(\max\left(\alpha_1,\beta_1\right),...,\max\left(\alpha_n,\beta_n\right)\right)$ and the greatest lower bound $\alpha\wedge\beta:=\left(\min\left(\alpha_1,\beta_1\right),...,\min\left(\alpha_n,\beta_n\right)\right)$ is the coordinate-wise minimum of the two vectors.  The relation $\theta\ge\eta$, when it holds, means that returns to effort are unambiguously greater in the model $\theta$ than they are in the model $\eta$.} is $\theta\vee\eta=\left(\alpha\vee\beta,m_1\wedge m_2\right)$ and whose meet is given by $\theta\wedge\eta=\left(\alpha\wedge\beta,m_1\vee m_2\right)$.
\end{definition}
The next Proposition shows that our chosen order structure $\left(\Theta,\leq\right)$ is the correct one, since each player's log-payoff now has increasing differences with respect to the model parameters.    
\begin{proposition}[Increasing Differences with Respect to Hardness]
Each player's log-payoff has increasing differences with respect to the hardness ($\geq$) of the parameter vector.  That is, given any fixed amount of extra effort $\Delta x_i>0$ for student $i$, the utility gain $\Delta\log U_i\left(x_i,x_{-i};\theta\right)=\log U_i\left(x_i+\Delta x_i,x_{-i};\theta\right)-\log U_i\left(x_i,x_{-i};\theta\right)$ is decreasing in the instructor's target mean $m$ and it is increasing in the ability vector $\left(\alpha_1,...,\alpha_n\right)$.  
\end{proposition}
Thus, in the sense of \cite{milgrom}, our supermodular game $\Gamma_n$ has been properly indexed, or parameterized (cf. with \cite{levin}), by the ordered set $\left(\Theta,\le\right)$.  In the sequel, such indexation will paramount for analyzing the (monotone) comparative statics (\cite{monotone}) of the students' equilibrium behavior.

\begin{proof}
First, the utility change $\Delta\log U_i\left(x;\theta\right)$ is unaffected by the abilities of the non-$i$ players;  clearly it is non-decreasing in the opposing ability vector $\alpha_{-i}$.  Now, $\Delta\log U_i$ is linear in player $i$'s own ability, and we have
\begin{equation}
\frac{\partial}{\partial\alpha_i}\left(\Delta\log U_i(x;\theta)\right)=\log\left(\frac{G_i(x_i+\Delta x_i,x_{-i};\theta)}{G_i(x;\theta)}\right)-\log\left(\frac{1-x_i-\Delta x_i}{1-x_i}\right)>0,
\end{equation}since the grade ratio on the left is $\geq1$ and the leisure ratio on the right is $<1$. 
\par
Next, in order to show that $\Delta\log U_i\left(x;\theta\right)$ is decreasing in $m$, we consider the grade ratio
\begin{equation}\label{graderatio}
m\mapsto\frac{x_i+\Delta x_i+\max\left(m-\overline{x}-\Delta x_i/n,0\right)}{x_i+\max\left(m-\overline{x},0\right)}
\end{equation}over the separate intervals $m\in\left[0,\overline{x}\right]$, $m\in\left[\overline{x},\overline{x}+\Delta x_i/n\right]$, and $m\in\left[\overline{x}+\Delta x_i/n,1\right]$, respectively.  For $m\leq\overline{x}$, the grade ratio (\ref{graderatio}) is a constant, i.e., it is non-increasing.  If $\overline{x}\leq m\leq\overline{x}+\Delta x_i/n$, then (\ref{graderatio}) equals $(x_i+\Delta x_i)/\left[x_i+\max\left(m-\overline{x},0\right)\right]$, which is decreasing in $m$.  Finally, if $m\geq\overline{x}-\Delta x_i/n$, then we have the function
\begin{equation}
m\mapsto1+\frac{\left(1-1/n\right)\Delta x_i}{x_i+\max\left(m-\overline{x},0\right)},
\end{equation}which decreases monotonically with the instructor's target mean.  Q.E.D.
\end{proof}

\begin{figure}[t]
\begin{center}
\includegraphics[height=200px]{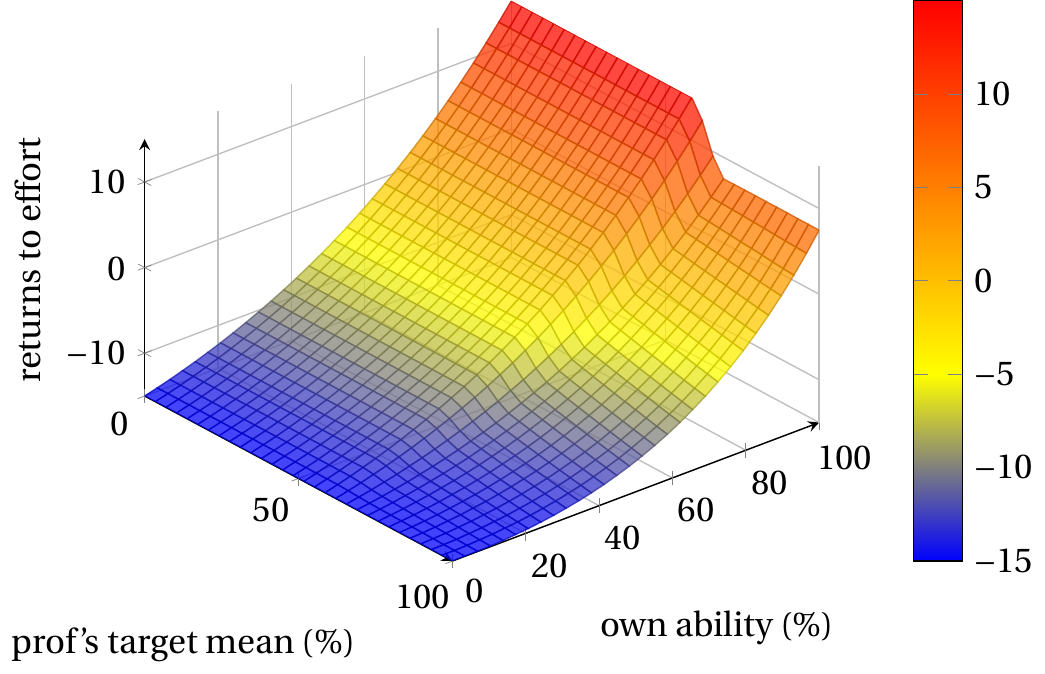}
\caption{\sc Proper indexation (\cite{milgrom,levin}) of the supermodular game (viz., increasing differences) by the poset $\left(\Theta,\le\right)$.  The returns to effort are increasing in each student's own ability $\alpha_i$, and decreasing in the teacher's target mean, $m$, which is a universal disincentive.  This illustration uses the parameters $\left(n,\overline{x}_{-i},x_i,\Delta x_i\right)=\left(2,60\%,60\%,15\%\right)$.  Player $i$'s utility gain $\Delta U_i$ is depicted on the vertical axis.}
\end{center}
\end{figure}

The following Corollary gives the general consequences (\cite{topkisbook,milgrom,yildizslides}) of the fact that each student's payoff has increasing differences with respect to his opponents' moves and also with respect to the hardness of the parameter vector.  In the sequel, these abstract results will be sharpened significantly, in so far as they apply to our concrete situation.

\begin{corollary}[Monotone Comparative Statics]
The greatest and least pure Nash equilibria (in the vector lattice $[0,1]^n$) are increasing in every ability parameter $\alpha_i$ and decreasing in the professor's target mean, $m$;  similarly for the extremal best responses

\begin{equation}\label{extremal}
\max\left(\argmax_{x_i\in[0,1]}U_i\left(x_i,x_{-i};\theta\right)\right)\text{ and }\min\left(\argmax_{x_i\in[0,1]}U_i\left(x_i,x_{-i};\theta\right)\right).
\end{equation}Due to the negative spillovers, the set of equilibria is totally ordered with respect to Pareto preference (\cite{milgrom,levin});  the minimum-effort pure Nash equilibrium Pareto dominates all the others.
\par
The greatest and least pure Nash equilibria are also, respectively, the greatest and least profiles of rationalizable strategies (\cite{milgrom,yildiz}).  If we iterate the greatest best responses (\ref{extremal})  on an initial seed of $\mathbb{1}:=(1,1,...,1)$, the resulting sequence converges to the greatest equilibrium point;  similarly, if we iterate the minimum best responses (\ref{extremal}) seeded by the zero vector, the resulting sequence converges to the least pure Nash equilibrium (\cite{milgrom}).

\end{corollary}

Before we proceed to find all Nash equilibria of the general $n$-person game, it is helpful to give a brief solution for the rational outcomes of a single-student course.  Here, the student has a ``bang-bang'' solution that oscillates between zero effort and his normal effort level of $x_1=\alpha_1$, depending on the particular value of the professor's target mean, $m$. 

\begin{example}[Baseline Behavior in a Single-Student Course]
If we have the smallest possible class size of $n:=1$, then $\overline{x}=x_1$, so that player 1's utility from the effort level $x_1$ is
\begin{equation}
U_1(x_1)=\left(x_1+\max\left(m-x_1,0\right)\right)^{\alpha_1}\left(1-x_1\right)^{1-\alpha_i}=\max\left(m,x_1\right)^{\alpha_1}\left(1-x_1\right)^{1-\alpha_1}.
\end{equation}Thus, in order to find $\argmax\limits_{x_1\in[0,1]} U_1(x_1)$, we must check the set of endpoints $\{0,1\}$, the point of non-differentiability ($m$), and the stationary point $\alpha_1$ in a no-curve optimum (cf. with \cite{meaning}).  Clearly, $x_1=1$ is sub-optimal, since the lack of leisure yields zero utility.  The sharp corner $x_1=m$ is inferior to $\alpha_1$, viz., 
\begin{equation}
U_1(m)=m^{\alpha_1}\left(1-m\right)^{\alpha_1}\leq\max\limits_{0\leq x_1\leq1}\left[x_1^{\alpha_1}(1-x_1)^{1-\alpha_1}\right]=\alpha_1^{\alpha_1}(1-\alpha_1)^{1-\alpha_1},
\end{equation}with the inequality being strict if $m\neq\alpha_1$.  Thus, the agent's optimal behavior in a single-student course is to put

\[x_1^*(\alpha_1,m)=\begin{cases} 
      \alpha_1 & \text{ if  }m<\alpha_1\left(1-\alpha_1\right)^{\frac{1-\alpha_1}{\alpha_1}}\\

  \left\{0,\alpha_1\right\} & \text{ if  }m=\alpha_1\left(1-\alpha_1\right)^{\frac{1-\alpha_1}{\alpha_1}}\\

0 & \text{ if  }m>\alpha_1\left(1-\alpha_1\right)^{\frac{1-\alpha_1}{\alpha_1}}.\\
    
   \end{cases}
\]

That is, if the curve is sufficiently generous, then the student will put zero effort;  otherwise, he will exert himself to the extent $\alpha_1$ that he normally does on an uncurved exam.  For the special cutoff value $m=\alpha_1\left(1-\alpha_1\right)^{\frac{1-\alpha_1}{\alpha_1}}$,  the student is just indifferent between these two extremes, and we get a pair of distinct optima.

\end{example}

\begin{figure}[t]
\begin{center}
\includegraphics[height=200px]{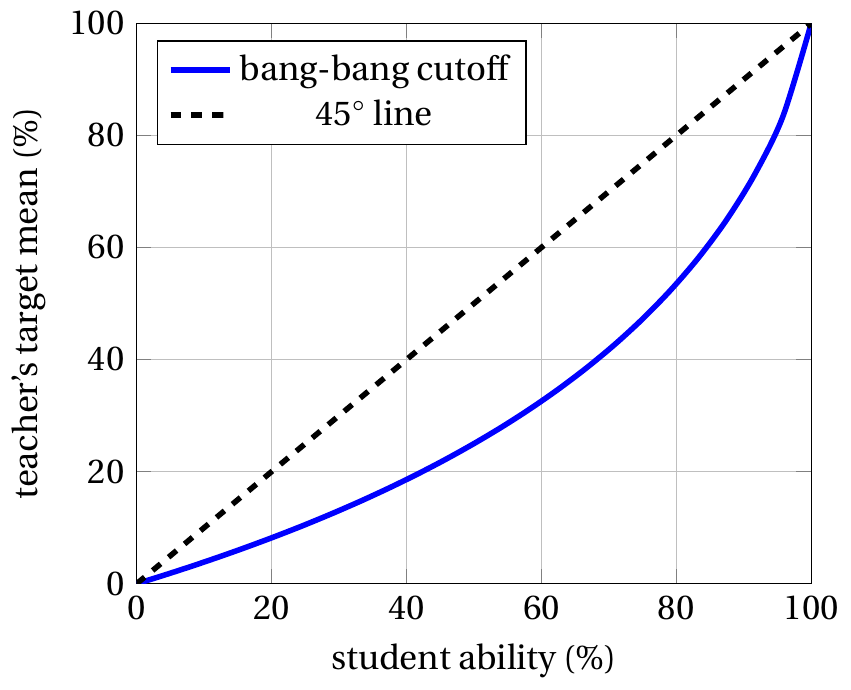}
\caption{\sc The bang-bang cutoff level $m=\alpha_1\left(1-\alpha_1\right)^{\frac{1-\alpha_1}{\alpha_1}}$ for the parameters of a single-student classroom.  Parameters $(\alpha_1,m)$ above this curve lead to zero student effort;  points below the curve generate full effort.  On the curve itself, the student is exactly indifferent between $x_1=0$ and $x_1=\alpha_1$.}
\end{center}
\end{figure}

\begin{figure}[t]
\begin{center}
\includegraphics[height=200px]{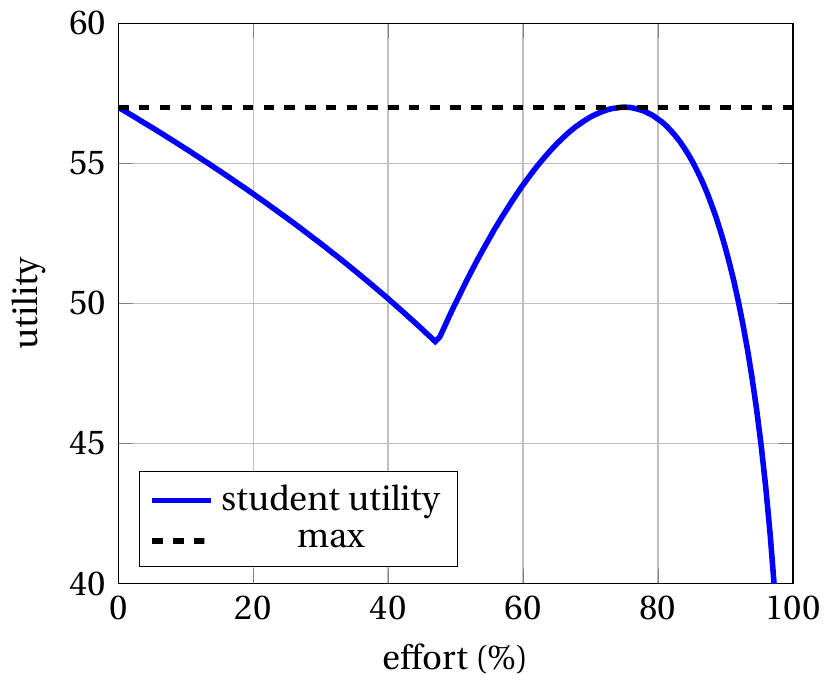}
\caption{\sc Multiple equilibria in a single-student classroom.  Here, we used the parameters $\alpha_1:=3/4$ and $m:=\nicefrac{3}{\left(4^{4/3}\right)}$.  The magnitude of this particular exam curve makes the student indifferent between zero effort and his usual effort level of $75\%$.}
\end{center}
\end{figure}

\section{Rational Behavior in Equilibrium.}

In order to unravel the exact, concrete structure of our contest $\Gamma_n$, it is necessary to identify various situations, from the standpoint of player $i$, whereby there is guaranteed to be a curve, or there is guaranteed to not be a curve, regardless of student $i$'s action.  In case of the former, we say that the curve has been ``made'' by $i$'s classmates $x_{-i}$, and in the latter, we say that the curve has been ``broken'' by $x_{-i}$.  If the curve is neither made nor broken by $i$'s opponents, then we say that $x_{-i}$ lies in the \textit{make-or-break region} for pupil $i$.  In this (tripartite) demarcation of the domain of kid $i$'s reaction correspondence, the payoff-relevant statistic is the opponents' sample mean $\overline{x}_{-i}$.  Thus, Proposition \ref{tripartite} below decomposes $[0,1]^{n-1}$ into a key triplet of convex polytopes.  In the sequel, student $i$'s most subtle and complicated economic behavior will occur over his or her make-or-break region;  in general, kid $i$ will then have to decide between distant pairs of critical reactions to the \textit{modus operandi} of his classmates.

\begin{proposition}[Making or Breaking the Exam Curve]\label{tripartite}
If 
\begin{equation}\label{breaking}
\overline{x}_{-i}\in\underbrace{\left(\frac{nm}{n-1},1\right]}_{\text{curve broken}},
\end{equation}then the non-$i$ players have guaranteed that there will be no curve, regardless of player $i$'s effort.  On the other hand, if 
\begin{equation}\label{making}
\overline{x}_{-i}\in\underbrace{\left[0,\frac{nm-1}{n-1}\right)}_{\text{curve made}},
\end{equation}then the non-$i$ players have guaranteed that there will be a curve, regardless of player $i$'s level of effort.  If
\begin{equation}\label{complement}
\overline{x}_{-i}\in\underbrace{\left[\frac{nm-1}{n-1},\frac{nm}{n-1}\right]}_{\text{make-or-break region}},
\end{equation}then there may or may not be a curve, depending on student $i$'s effort level.  Specifically, there is a curve if and only if $x_i$ is below the cutoff value 
\begin{equation}
\hat{x_i}:=\underbrace{nm-(n-1)\overline{x}_{-i}}_{\text{value of }x_i\text{ that makes }\overline{x}=m},
\end{equation}otherwise there is no curve.  The critical region (\ref{complement}) always contains the instructor's target mean $m$, and it collapses to $\{m\}$ as $n\to\infty$.
\end{proposition}

\begin{proof}
To derive the curve-breaking condition (\ref{breaking}), assume that player $i$ puts minimum effort ($x_i=0$).  Then the curve will be broken if and only if $S_{-i}/n>m$, which is equivalent to the stated condition $\overline{x}_{-i}>nm/(n-1)$.  Similarly, assume that player $i$ exerts maximum effort ($x_i=1$).  Then, an exam curve will be in effect if and only if $\left(1+S_{-i}\right)/n<m$, which is equivalent to the fact that $\overline{x}_{-i}<(nm-1)/(n-1)$, as promised.  Finally, when the average effort of the non-$i$ players lies in the interval (\ref{complement}), then the sign of $m-\overline{x}$ is under player $i$'s control, and can go either way.  The cutoff value $\hat{x_i}$ of $x_i$ is specified by the equation $(S_{-i}+\hat{x_i})/n=m$, or, equivalently, $\hat{x_i}=nm-\left(n-1\right)\overline{x}_{-i}$.  If player $i$'s effort is less than this cutoff value, then we have $\overline{x}<m$, and there will be an exam curve.  If $x_i$ exceeds the cutoff, then we will have $\overline{x}>m$, and the curve will be broken.  Finally, in (\ref{complement}), we have $\left(nm-1\right)/\left(n-1\right)=m+\left(m-1\right)/\left(n-1\right)<m$, since $m-1$ is negative.  Taking the limit of the endpoints of the segment (\ref{complement}) as $n\to\infty$, we obtain the degenerate interval $[m,m]=\left\{m\right\}$.  Q.E.D.
\end{proof}

\begin{figure}[t]
\begin{center}
\includegraphics[height=200px]{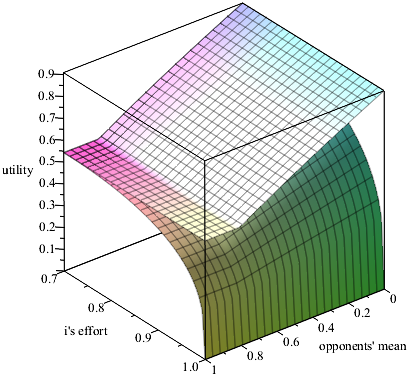}
\caption{\sc Strict dominance of the strategies in the interval $(\alpha_i,1]$ by $\alpha_i$.  The upper surface is the graph of $\left(x_i,\overline{x}_{-i}\right)\mapsto U_i(\alpha_i,\overline{x}_{-i})$, and the lower surface is the graph of $U_i(x_i,\overline{x}_{-i})$ over the rectangle $[\alpha_i,1]\times[0,1]$.  This illustration used the parameters $\left(n,m,\alpha_1\right):=(18,80\%,70\%)$.}
\end{center}
\end{figure}

The next Proposition says that all effort levels above the usual optimum $x_i=\alpha_i$ (that would obtain for player $i$ in an uncurved economics course) are strictly dominated by $\alpha_i$ itself.  That is, since the professor's curving scheme is an incentive to reduce effort, there is never any reason to give more than $\alpha_i$.  Of course, there are ample reasons to give less:  the padding of the curve means that agent $i$ can decrease his effort somewhat, and still wind up with a higher grade than the $\alpha_i$ he would normally receive.  In a small class, the curve furnishes a sizeable ``refund'' for dumping on the exam.  Say, for $n:=2$ students, every $2\%$ that you lose tanks the class average by $1\%$, which is paid back to you by the exam curve.  This latter mechanic, which leads to a jump in each player's reaction correspondence, gets erased in the limit as $n\to\infty$.

\begin{proposition}[Dominated Strategies]\label{dom}
For every player $i$, the strategies in the interval $(\alpha_i,1]$ are all strictly dominated by $\alpha_i$.  Thus, in analyzing the game $\Gamma_n$, we may restrict our attention to the box $\bigtimes\limits_{i=1}^n[0,\alpha_i]$.
\end{proposition}
\begin{proof}
Let $G_i(x_i,x_{-i})$ denote player $i$'s grade when he plays $x_i$, and let $G_i(\alpha_i,x_{-i})$ denote his grade when he plays $\alpha_i$ against $x_{-i}$.  We will apply the general theorem of the arithmetic and geometric means (the AGM inequality, cf. with \cite{berge}) to the utility ratio
\begin{equation}\label{agm}
\frac{U_i(x_i,x_{-i})}{U_i(\alpha_i,x_{-i})}=\left(\frac{G_i(x_i,x_{-i})}{G_i(\alpha_i,x_{-i})}\right)^{\alpha_i}\left(\frac{1-x_i}{1-\alpha_i}\right)^{1-\alpha_i}.
\end{equation}Note that the grade ratio $G_i(x_i,x_{-i})/G_i(\alpha_i,x_{-i})$ is at least one, since each player's grade $G_i$ is non-decreasing in his own effort, and we have the hypothesis that $x_i>\alpha_i$.  On the other hand, the leisure ratio $(1-x_i)/(1-\alpha_i)$ is strictly less than one, so that the two factors that appear in the geometric mean (\ref{agm}) are distinct numbers.  Thus, we have the (strict) AGM inequality
\begin{equation}
\frac{U_i(x_i,x_{-i})}{U_i(\alpha_i,x_{-i})}<\alpha_i\frac{G_i(x_i,x_{-i})}{G_i(\alpha_i,x_{-i})}+1-x_i.
\end{equation}Accordingly, it suffices to prove the relation
\begin{equation}\label{suffices}
\alpha_i\frac{G_i(x_i,x_{-i})}{G_i(\alpha_i,x_{-i})}\leq x_i
\end{equation}in order to establish the fact that $\alpha_i$ strictly dominates $x_i$.  We will show below that
\begin{equation}
\max_{\alpha_i\in[0,x_i]}\frac{\alpha_i}{G_i(\alpha_i,x_{-i})}=\frac{x_i}{G_i(x_i,x_{-i})},
\end{equation}at which point (\ref{suffices}) will have been demonstrated in earnest.  We have
\begin{equation}\label{monotonic}
\frac{\alpha_i}{G_i\left(\alpha_i,x_{-i}\right)}=\max\left(\frac{n-1}{n}+\alpha_i^{-1}\left(m-\frac{n-1}{n}\overline{x}_{-i}\right),1\right)^{-1},
\end{equation}on account of the fact that $\max\left(\bullet,\bullet\right)$ is positively homogeneous of degree one.  According to (\ref{monotonic}), then, the function $\alpha_i\mapsto\alpha_i/G_i(\alpha_i,x_{-i})$ is either monotonically decreasing (if $m\leq(n-1)\overline{x}_{-i}/n$) or monotonically increasing (if $m\geq\left(n-1\right)\overline{x}_{-i}/n)$.  Thus, the maximum value of (\ref{monotonic}) for $\alpha_i\in[0,x_i]$ must occur at one of the endpoints $\{0,x_i\}$.  Hence, the expression (\ref{monotonic}) is majorized by $x_i/G_i(x_i,x_{-i})$, and the Proposition is proved.
\end{proof}
Thanks to Propositions \ref{tripartite} and \ref{dom}, we have the following Corollary, which gives some simple necessary or sufficient conditions for whether or not there will be an exam curve in equilibrium.

\begin{corollary}[Basic Conditions for Curved \& Uncurved Equilibria]\label{basic}
A necessary condition for a no-curve equilibrium is that $\overline{\alpha}\geq m$;  and a sufficient condition for the existence of a no-curve equilibrium is that $\overline{\alpha}>m+(1/n)\max\limits_{1\leq i\leq n}\alpha_i$.  Thus, $\overline{\alpha}<m$ will suffice for an equilibrium that has an exam curve;  the condition $\overline{\alpha}\le m+(1/n)\max\limits_{1\leq i\leq n}\alpha_i$ is necessary for the existence of a curved equilibrium.
\end{corollary}
\begin{proof}
Assume that $x\in[0,1]^n$ is an equilibrium point for which the curve is broken ($\overline{x}\geq m$).  Then, since dominated strategies cannot be played in equilibrium, we must have $x_i\leq\alpha_i$ for all $i$, so that $\overline{\alpha}\geq\overline{x}\geq m$.  Now, let $\alpha_{-i}\in[0,1]^{n-1}$ denote the vector of ability parameters of the non-$i$ players.  In the event that $\overline{\alpha}_{-i}$ breaks the curve for all $i$, then each player $i$'s best response is to himself put $x_i=\alpha_i$, since the curve is guaranteed to be broken, regardless of his own action.  Thus, in order to generate a curved equilibrium, it suffices to have $\overline{\alpha}_{-i}>nm/(n-1)$ for all $i$, which is equivalent to $\overline{\alpha}>m+(1/n)\max\limits_{1\leq i\leq n}\alpha_i$.  Taking the contrapositive of these respective conditions for an uncurved equilibrium, we obtain their stated counterparts for an equilibrium that features an exam curve.  Q.E.D.
\end{proof}

Based on the foregoing (tripartite) decomposition of the action profiles $x_{-i}\in[0,1]^{n-1}$ of student $i$'s opponents, we have the following Lemma, which gives a fundamental expression for each player's best response correspondence.

\begin{lemma}[Basic Structure of the Reaction Correspondence]\label{br}
Player $i$'s best response correspondence, $BR_i(x_{-i}):=\argmax\limits_{x_i\in[0,1]}U_i\left(x_i,x_{-i}\right)$ is given by the following piecewise formula:\footnote{One or more of the intervals in this piecewise correspondence may turn out to be empty, depending on our precise location $(\alpha,m)\in\Theta$ in the parameter space.  Say, if $m\geq1-1/n$, then the no-curve region $\left(nm/\left(n-1\right),1\right]$ will be empty for all players.  Similarly, if $nm/(n-1)<\alpha_i/(1-\alpha_i)$, then player $i$'s no-show region will be empty.  However, the best response formula given in the text is still correct, as it simply asserts that if $\overline{x}_{-i}$ belongs to such-and-such segment, then the set of all best responses amounts to such-and-such.  If any of the intervals in the piecewise formula turn out to be empty, then the assertion is vacuously true.}

\[\boxed{ BR_i\left(x_{-i}\right)=\begin{cases}
      \alpha_i & \text{ if  }\frac{nm}{n-1}<\overline{x}_{-i}\leq1\text{   (curve broken)}\\

\argmax\limits_{x_i\in C(x_{-i})}U_i(x_i,x_{-i})& \text{ if }\frac{nm-1}{n-1}\leq\overline{x}_{-i}\leq\frac{nm}{n-1}\text{  (make-or-break region)}\\
      
      \alpha_i-\left(1-\alpha_i\right)\left(\frac{nm}{n-1}-\overline{x}_{-i}\right) & \text{ if } \frac{nm}{n-1}-\frac{\alpha_i}{1-\alpha_i}\leq\overline{x}_{-i}<\frac{nm-1}{n-1}\text{  (curve made)}\\
      
0 & \text{ if }0\leq\overline{x}_{-i}\leq\frac{nm}{n-1}-\frac{\alpha_i}{1-\alpha_i}\text{ (no-show region),}      
   \end{cases}
}\]where $C(x_{-i})$ is the set of three critical points
\begin{equation}\boxed{
C(x_{-i}):=\left\{0,\alpha_i-(1-\alpha_i)\left(\frac{nm}{n-1}-\overline{x}_{-i}\right),\alpha_i\right\}\cap[0,1].}
\end{equation}
\end{lemma}

\begin{proof}
If $\overline{x}_{-i}>nm/(n-1)$, meaning that the effort level of $i$'s opponents is so high as to guarantee that the curve is broken, regardless of $x_i$, then student $i$'s best play is to put his normal effort $x_i=\alpha_i$, as we have seen above.  On the other hand, if $\overline{x}_{-i}<(nm-1)/(n-1)$, then the effort level of $i$'s opponents is low enough to guarantee that the curve is made, regardless of $i$'s behavior.  In this happenstance, a best respondent must optimize

\begin{equation}\label{objective}
\max_{0\leq x_i\leq 1}\left(m+\frac{n-1}{n}\left(x_i-\overline{x}_{-i}\right)\right)^{\alpha_i}\left(1-x_i\right)^{1-\alpha_i}.
\end{equation}Taking the $\log$ of the objective (\ref{objective}), and solving the first order condition 
\begin{equation}
\frac{\partial}{\partial x_i}\log\left(U_i\left(x_i,x_{-i}\right)\right)=0
\end{equation}for an interior optimum, we obtain
\begin{equation}\label{interior}
x_i^*=\alpha_i-(1-\alpha_i)\left(\frac{nm}{n-1}-\overline{x}_{-i}\right).
\end{equation}The program (\ref{objective}) cannot have a solution at the corner $x_i=1$, since it corresponds to no leisure time and, accordingly, zero utility.  However, (\ref{objective}) will have the corner solution $x_i=0$ (meaning zero effort) precisely when the formula for $x_i^*$ given in (\ref{interior}) is $\leq0$, viz., when
\begin{equation}
\overline{x}_{-i}\leq\frac{nm}{n-1}-\frac{\alpha_i}{1-\alpha_i}.
\end{equation}Finally, we have the problem of optimizing agent $i$'s utility over the make-or-break region, whereby the existence or non-existence of the exam curve hinges on the particular behavior of kid $i$.  Just as in the single-student example from the prequel, we must consider all endpoints, points of non-differentiability, and all points where the derivative $\partial U_i/\partial x_i$ may be zero.  Thus, we have the five critical points
\begin{equation}
\left\{0,1,\underbrace{nm-(n-1)\overline{x}_{-i}}_{\hat{x}_i\text{, kink in }i\text{`s utility}},\alpha_i-(1-\alpha_i)\left(\frac{nm}{n-1}-\overline{x}_{-i}\right),\alpha_i\right\}.
\end{equation}As we have remarked above, the corner $x_i=1$ can never be a solution, since it gives zero utility (viz., $\alpha_i$ is better).  The kink point, $\hat{x}_i:=nm-(n-1)\overline{x}_{-i}$, is not an optimum unless it coincides with $\alpha_i$, i.e., $U_i\left(\hat{x}_i,x_{-i}\right)=\hat{x}_i^{\alpha_i}\left(1-\hat{x}_i\right)^{1-\alpha_i}<\alpha_i^{\alpha_i}\left(1-\alpha_i\right)^{1-\alpha_i}\leq U_i(\alpha_i,x_{-i})$ if $\hat{x}_i\neq\alpha_i$.  Here, we have used the security level $U_i(\alpha_i,x_{-i})\geq\alpha_i^{\alpha_i}\left(1-\alpha_i\right)^{1-\alpha_i}$, which inequality is true on account of the fact that $\alpha_i+\max\left(m-\overline{x},0\right)\geq\alpha_i$.  Thus, we are left with the set $C(x_{-i})$ of three critical points that were given in the statement of the Lemma.  Q.E.D.
\end{proof}
Note that the zero-effort interval $\left[0,nm/(n-1)-\alpha_i/(1-\alpha_i)\right]$ will be empty if and only if
\begin{equation}\label{zeroeffort}
\frac{nm}{n-1}<\frac{\alpha_i}{1-\alpha_i},
\end{equation}or equivalently, when 
\begin{equation}
\alpha_i>1-\left(1+\frac{nm}{n+1}\right)^{-1}.
\end{equation}In a large classroom, as $n\to\infty$, this condition converges to $\alpha_i\geq m/(m+1)$, i.e., student $i$ will have positive effort in any best response, regardless of class size, as long as his ability parameter $\alpha_i$ is sufficiently high. 

\begin{example}
If a student's quality parameter $\alpha_i$ is just $50\%$ in a course with $n:=18$ students and a professor's target mean of $m:=75\%$, then the condition (\ref{zeroeffort}) amounts to the fact that $1>0.794$.  As $n\to\infty$, positive effort is guaranteed for all students whose ability parameter $\alpha_i$ satisfies
\begin{equation}
\alpha_i\geq\frac{m}{m+1}=42.9\%,
\end{equation}so that all A, B, C, and D students are guaranteed to have positive effort in any best response, regardless of the number of students who sit for the exam.
\end{example}
\begin{proposition}[Asymptotic Best Responses]\label{myprop}
As the class size $n\to\infty$, each player's make-or-break region becomes negligible, and student $i$'s limiting best response function $BR_i^\infty\left(x_{-i}\right)$ consists in the isotone function

\[ BR_i^\infty\left(x_{-i}\right)=\begin{cases} 
      \alpha_i & \text{ if  }m\leq\overline{x}_{-i}\leq1\text{   (curve broken, full effort)}\\

      \alpha_i-\left(1-\alpha_i\right)\left(m-\overline{x}_{-i}\right) & \text{ if } m-\frac{\alpha_i}{1-\alpha_i}\leq\overline{x}_{-i}\leq m\text{  (curve, partial effort)}\\
      
0 & \text{ if }0\leq\overline{x}_{-i}\leq m-\frac{\alpha_i}{1-\alpha_i}\text{ (curve, no effort).}      
   \end{cases}
\]

\end{proposition}
The proof of Proposition \ref{myprop} is immediate, by taking $n\to\infty$ in the best response curve $BR_i^{(n)}(\bullet)$ that was constructed in Lemma \ref{br}.  For large $n$, the graph of the correspondence $BR^{(n)}_i:[0,1]^{n-1}\rightrightarrows[0,1]$ will be practically indistinguishable from that of the function $BR_i^{\infty}(x_{-i})$ that is given in the proposition.

\begin{figure}[t]
\begin{center}
\includegraphics[height=200px]{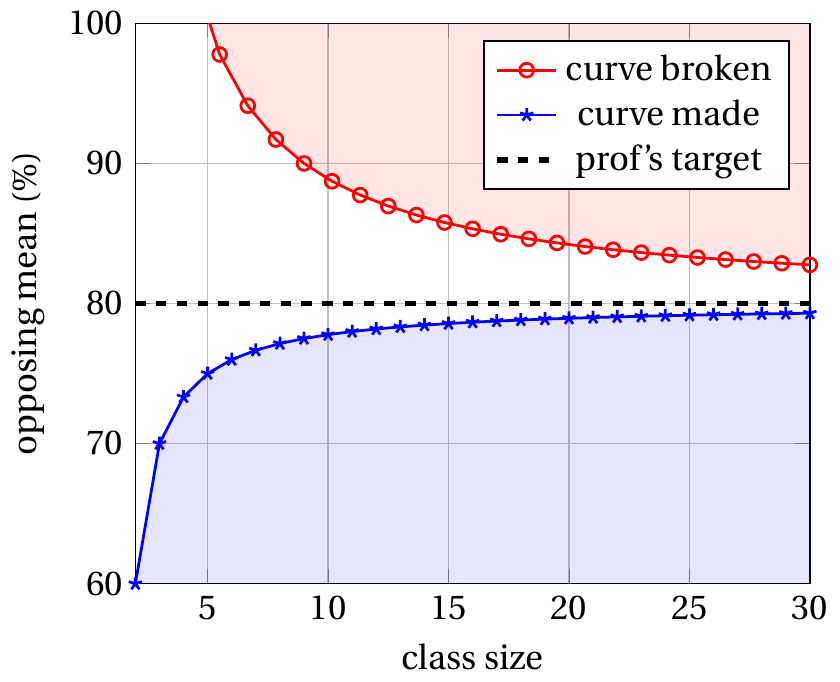}
\caption{\sc The Make-or-Break Region $\overline{x}_{-i}\in\left[\left(nm-1\right)/(n-1),nm/(n-1)\right]$ for different class sizes $n\in\left[2,30\right]\cap\mathbb{Z}$, assuming that $m:=80\%$ is the professor's target mean.  If the average effort of student $i$'s opponents lies in this interval, then student $i$'s raw score can make or break the curve.  If $\overline{x}_{-i}$ lies outside of this region, then the curve has already been made or broken by the ex-$i$ players.}
\end{center}
\end{figure}

The next Proposition uses the supermodularity of $\Gamma_n$ to show that all of player $i$'s strategies that lie to the left of his best reply to the zero vector are strictly dominated;  this explains why they are never best responses in the reaction correspondence that was given above. 

\begin{proposition}\label{dominated2}
For each player $i$, all the pure strategies in the interval\footnote{This interval may or may not be empty, depending on the parameters of the model.  If it is empty, then the Proposition is vacuously true.} 
\begin{equation}
\left[0,\alpha_i-\left(1-\alpha_i\right)\frac{nm}{n-1}\right)=\left[0,BR_i(\mathbf{0})\right)
\end{equation}are strictly dominated by $\alpha_i-\left(1-\alpha_i\right)nm/\left(n-1\right)$, which is the best response to the zero vector $x_{-i}=\mathbf{0}_{n-1}=\left(0,...,0\right)$.
\end{proposition}
\begin{proof}
Assume that $x_i\in\left[0,BR_i(\mathbf{0})\right)$, and form the difference (cf. with \cite{yildizslides})
\begin{equation}\label{mydiff}
\log U_i\left(BR_i(\mathbf{0}),x_{-i}\right)-\log U_i\left(x_i,x_{-i}\right)\geq\log U_i\left(BR_i(\mathbf{0}),\mathbf{0}\right)-\log U_i\left(x_i,\mathbf{0}\right)>0.
\end{equation}The first inequality ($\ge$) in (\ref{mydiff}) obtains from the fact that $\log\left( U_i(\bullet)\right)$ has increasing differences, and $x_{-i}\geq\mathbf{0}$;  the second inequality ($>$) obtains from the hypothesis that $x_i$ is not the best response to $\mathbf{0}$.  Thus, we must have $U_i\left(BR_i(\mathbf{0}),x_{-i}\right)>U_i\left(x_i,x_{-i}\right)$ for all $x_{-i}\in[0,1]^{n-1}$, so that $x_i$ is strictly dominated by $\alpha_i-(1-\alpha_i)nm/\left(n-1\right)$, and the Proposition is proved.

\end{proof}

The following useful Theorem gives exact formulas for the equilibrium behavior of $n$ pupils who all give positive effort, and yet such effort is coordinated (within the scope of our classical, non-cooperative framework) for the sake of generating a mutually beneficial exam curve.  Such coordination will allow all of the students to decrease their efforts (and enjoy more leisure time), although the exact amount of this decrease will be idiosyncratic to the individual ability levels.  As we will see below, the equilibrium behavior becomes particularly elegant and simple in the limit as the class size becomes large.

\begin{theorem}[Equilibrium Exam Curve]\label{nstudent}
In a (curved, interior) pure Nash equilibrium (\cite{nash,nash51}) of an $n$-student classroom, the average raw score $\overline{x}^*$ on the exam is given by the formula
\begin{equation}\boxed{
\overline{x}^*=1-\frac{nm}{n-1}\left(\frac{1}{\hat{\alpha}_n}-1\right),}
\end{equation}where the parameter
\begin{multline}\label{parameter}
\hat{\alpha}_n:=n\left(1-\left(\sum_{i=1}^n\frac{1}{n-\alpha_i}\right)^{-1}\right)\\
=n-\text{HarmonicMean}\left(n-\alpha_1,n-\alpha_2,...,n-\alpha_n\right)\in\left[\min_{1\leq i\leq n}\alpha_i,\max_{1\leq i\leq n}\alpha_i\right]
\end{multline}is the proper measurement of class ability.  Player $i$'s raw exam score $x_i^*$ consists in the expression
\begin{equation}\label{equilibrium}\boxed{
x_i^*=\frac{(n-1)\alpha_i-n\left(1-\alpha_i\right)\left(m-\overline{x}^*\right)}{n-\alpha_i}.}
\end{equation}
\end{theorem}
\begin{proof}
Let $S:=\sum\limits_{i=1}^nx_i=S_{-i}+x_i$ denote the aggregate classroom effort.  In the curved best response condition (\ref{interior}), substituting $\overline{x}_{-i}=\left(S-x_i\right)/\left(n-1\right)$ gives us
\begin{equation}
x_i=\alpha_i-\left(1-\alpha_i\right)\frac{nm-S+x_i}{n-1},
\end{equation}so that, solving for $x_i$ in terms of $S$, we have 
\begin{equation}\label{summing}
x_i=\frac{(n-1)\alpha_i-\left(1-\alpha_i\right)\left(nm-S\right)}{n-\alpha_i}=\frac{(n-1)\alpha_i-(1-\alpha_i)nm}{n-\alpha_i}+\frac{1-\alpha_i}{n-\alpha_i}\cdot S.
\end{equation}Summing (\ref{summing}) over all students $i$, and solving for $S$, we obtain
\begin{equation}\label{simplify}
S=\frac{\sum\limits_{i=1}^n\frac{(n-1)\alpha_i-(1-\alpha_i)nm}{n-\alpha_i}}{1-\sum\limits_{i=1}^n\frac{1-\alpha_i}{n-\alpha_i}}.
\end{equation}In order to simplify (\ref{simplify}), we let
\begin{equation}\label{auxsums}
S_1:=\sum_{i=1}^n\frac{\alpha_i}{n-\alpha_i}\text{ and }S_2:=\sum_{i=1}^n\frac{1}{n-\alpha_i},
\end{equation}so that
\begin{equation}\label{sumexpr}
S=\frac{\left(n-1+nm\right)S_1-nmS_2}{1+S_1-S_2}.
\end{equation}Now, using the fact that $S_1=n(S_2-1)$, we get
\begin{multline}\label{avgeffort}
\overline{x}^*=\frac{S}{n}=\frac{1}{n-1}\left(n-1+nm-\frac{mS_2}{S_2-1}\right)=1+\frac{1}{n-1}\left(nm-\frac{m}{1-S_2^{-1}}\right)\\
=1-\frac{nm}{n-1}\left(\frac{1}{\hat{\alpha_n}}-1\right),
\end{multline}as promised.  Finally, putting $S^*=n\overline{x}^*$ in (\ref{summing}), and simplifying, we obtain the lovely expression (\ref{equilibrium}) for player $i$'s equilibrium behavior.  Note that the harmonic mean $H_n$ of the numbers $\left(n-\alpha_1,n-\alpha_2,...,n-\alpha_n\right)$ must lie in the interval 
\begin{equation}
\left[\min_{1\leq i\leq n}(n-\alpha_i),\max_{1\leq i\leq n}(n-\alpha_i)\right]=\left[n-\max\limits_{1\leq i\leq n}\alpha_i,n-\min\limits_{1\leq i\leq n}\alpha_i\right],
\end{equation}so that $\min\limits_{1\leq i\leq n}\alpha_i\leq n-H_n\leq\max\limits_{1\leq i\leq n}\alpha_i$.  Q.E.D.
\end{proof}

\begin{figure}[t]
\begin{center}
\includegraphics[height=200px]{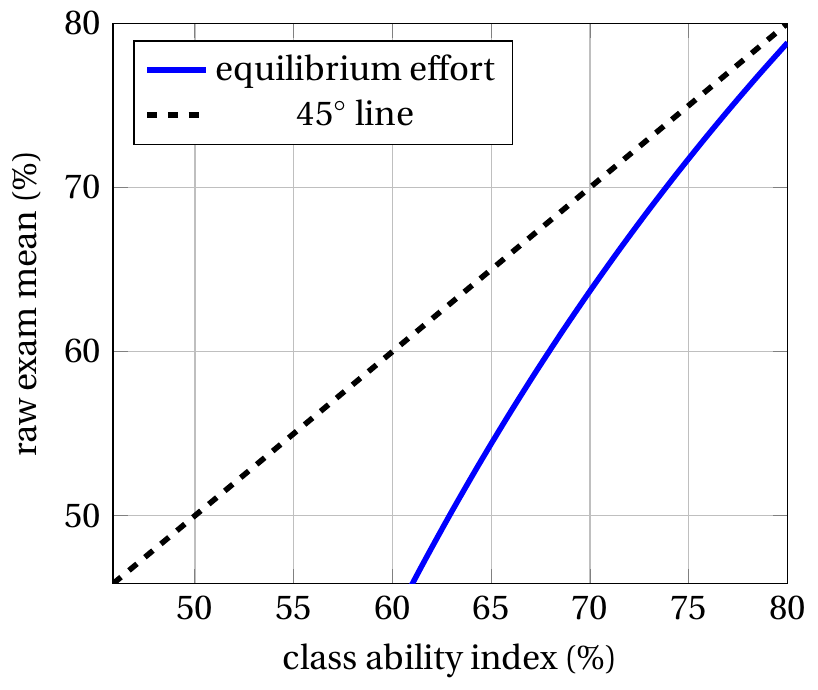}
\caption{\sc Equilibrium classroom effort $\overline{x}^*$ for different ability indices $\hat{\alpha}\leq m$, assuming a class size of $n:=18$ students and a professor's target mean of $m:=80\%$.}
\end{center}
\end{figure}

Note well that the formula (\ref{parameter}) for $\hat{\alpha}_n$ constitutes a legitimate mean (distinct from, say, the arithmetic mean) of the student quality parameters $\left(\alpha_i\right)_{i=1}^n$.  That is, we first subtract all the $\alpha_i$ from $n$;  we then take the harmonic mean $H_n$ of the resulting sequence of numbers;  finally, we subtract that number from $n$ in order to ``undo'' the initial operation $\alpha_i\mapsto n-\alpha_i$ that was fed into the harmonic mean.  Say, if the students' abilities $\alpha_i\equiv\alpha$ are all equal, then our averaging process gives $\hat{\alpha}=\alpha$, which is a sensible result.  The measurement $\hat{\alpha}$ is increasing in each parameter $\alpha_i$, viz., if $\alpha_i$ increases, then the numbers $\left(n-\alpha_1,...,n-\alpha_n\right)$ all decrease, so that their harmonic mean $H_n$ decreases, whence $n-H_n$ increases.  Since the harmonic mean $H_n(\bullet)$ is concave, our game-theoretic measure $\hat{\alpha}\left(\alpha_1,...,\alpha_n\right)$ is a convex function of the students' Cobb-Douglas parameters, since we have substituted affine functions $\alpha_i\mapsto n-\alpha_i$ into $H_n(\bullet)$, and then taken the opposite (cf. with \cite{boyd}).  The student ability index is a symmetric function of the $\alpha_i$, viz., we have $\hat{\alpha}\left(\alpha_{j_1},\alpha_{j_2}...,\alpha_{j_n}\right)\equiv\hat{\alpha}\left(\alpha_1,\alpha_2,...,\alpha_n\right)$ for any permutation $\left(j_1,j_2,...,j_n\right)$ of the students $\left\{1,2,...,n\right\}$.

\begin{figure}[t]
\begin{center}
\includegraphics[height=200px]{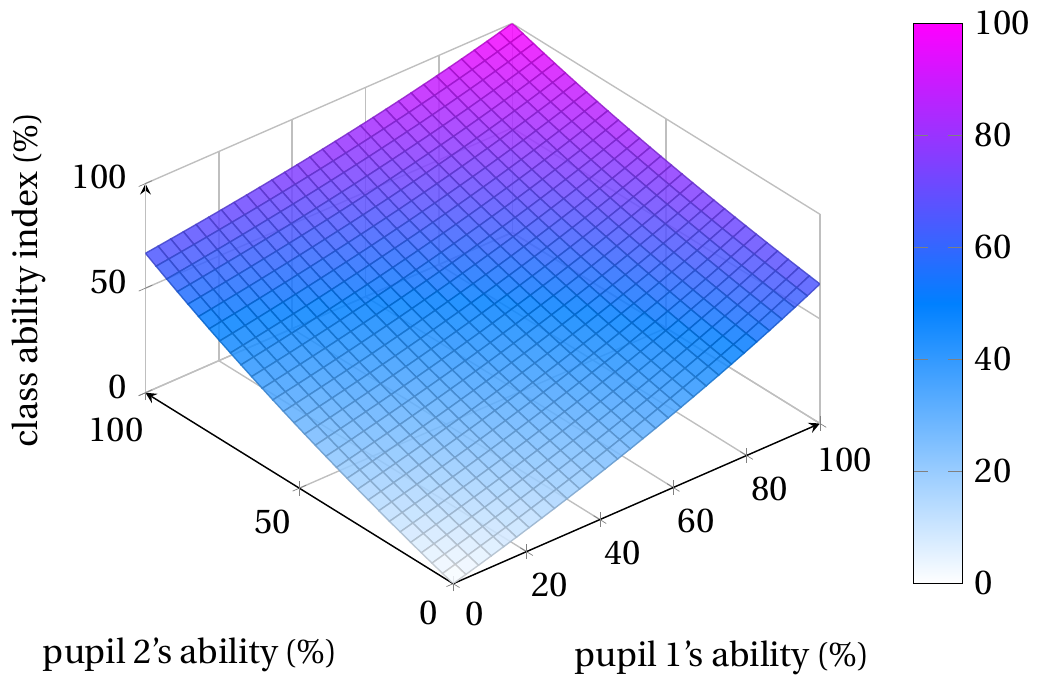}
\caption{\sc Illustration of the equilibrium ability measure $\hat{\alpha}(\bullet)$, for the case of $n:=2$ students.  Here, the game-theoretic ability index consists in the formula $\hat{\alpha}\left(\alpha_1,\alpha_2\right)=2\left(\alpha_1+\alpha_2-\alpha_1\alpha_2\right)/\left(4-\alpha_1-\alpha_2\right)$, which is an increasing, continuous, symmetric, convex function of $\left(\alpha_1,\alpha_2\right)$.  We have the diagonal values $\hat{\alpha}(z,z)\equiv z$, and the bounds $\min\left(\alpha_1,\alpha_2\right)\leq\hat{\alpha}\left(\alpha_1,\alpha_2\right)\leq\max\left(\alpha_1,\alpha_2\right)$. }
\end{center}
\end{figure}

\begin{corollary}[The Grade Inflation Theorem]\label{mycor}
As the class size $n\to\infty$, the average raw (effort) score on the exam converges to
\begin{equation}
\overline{x}^{*}_\infty=1-m\left(\frac{1}{\hat{\alpha}_\infty}-1\right),
\end{equation}where
\begin{equation}
\hat{\alpha}_\infty:=\lim_{n\to\infty}\hat{\alpha}_n
\end{equation}is the game-theoretic ability index of the student population.  Kid $i$'s equilibrium effort converges to 
\begin{equation}\boxed{
	x_i^{*\infty}=\alpha_i-(1-\alpha_i)\left(\frac{m}{\hat{\alpha}_\infty}-1\right);}
\end{equation}thus, the exam curve will asymptotically create a net change in effort $x_i^{*\infty}-\alpha_i$ that is directly proportional to $-(1-\alpha_i)$.  That is, the strongest students in the class will have the lowest absolute decrease in their effort levels.  In equilibrium, all students in the curved course will increase their leisure time by the same percentage:
\begin{equation}\boxed{
\frac{L_i^{*\infty}}{1-\alpha_i}=\frac{m}{\hat{\alpha}_\infty}.}
\end{equation}Finally, all students' grades  will get inflated by a factor of $m/\hat{\alpha}_\infty$ relative to an uncurved situation, viz., we have
\begin{equation}\boxed{
\frac{G_i^{*\infty}}{\alpha_i}=\frac{m}{\hat{\alpha}_\infty},}
\end{equation}where $G_i^{*\infty}$ is the asymptotic grade of kid $i$ as the class size becomes large.  Thus, the professor can deduce student $i$'s true ability parameter $\alpha_i$ by using the formula $\alpha_i=\hat{\alpha}_\infty G_i^{*}/m$, where $G_i^{*}$ is student $i$'s curved exam grade.
\end{corollary}
The proof of Corollary \ref{mycor} consists in taking $n\to\infty$ in the $n$-student equilibrium that was derived in Theorem \ref{nstudent}, and simplifying.  Here, we are tacitly imposing an asymptotic stability condition on quality of the student population, e.g., we assume that the limit
\begin{equation}\boxed{
\hat{\alpha}_\infty=\lim_{n\to\infty}n\left(1-\left(\sum\limits_{i=1}^n\frac{1}{n-\alpha_i}\right)^{-1}\right)}
\end{equation}exists.  Thus, we have a general ``grade inflation factor'' of $m/\hat{\alpha}_\infty$;  although all students receive the same number of free points $m-\overline{x}_\infty^*$ from the curve, the weaker students decrease their effort more (in absolute terms) than do the stronger students.  This happens on account of the fact that all students increase their leisure time by a fixed percentage ($=m/\hat{\alpha}_\infty-1$):  the bottom students already take high number of leisure hours, so that the absolute change in their leisure time is high.  On the other hand, the try-hard students are increasing their leisure by the same percentage, but from a very low base.  For instance, if $m=80\%$ and $\hat{\alpha}_\infty=70\%$, then all students will take $14\%$ more leisure time than they did before, and all students' equilibrium grades will get multiplied by a factor of $1.14$.  Thus, a student who would normally score $90\%$ should receive $103\%$ in equilibrium, net of the curve, for a gain of $13$ percentage points.  On the other hand, a C student whose ability parameter is $70\%$ will wind up with $80\%$ for a net gain of only $10$ percentage points.
\par
Having derived the interior equilibrium with an exam curve, we proceed to study the fine-grained properties and character of each player's best response correspondence $BR_i:[0,1]^{n-1}\rightrightarrows[0,1]$ over the make-or-break region.  We will require the following Lemma in order to remove $x_i=0$ from the set $C(x_{-i})$ of critical points that features prominently in Lemma \ref{br}.

\begin{figure}[t]
\begin{center}
\includegraphics[height=200px]{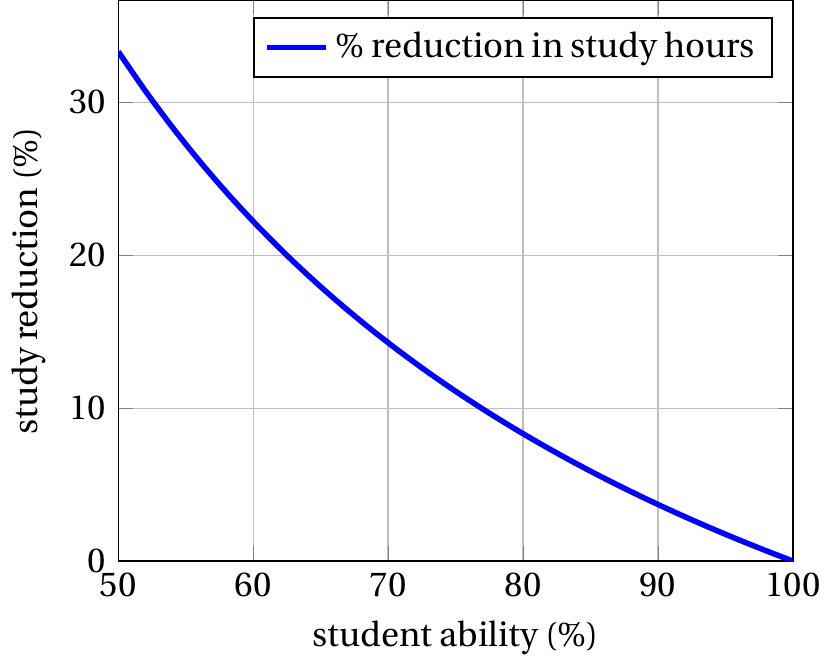}
\caption{\sc The equilibrium percentage reduction in total study hours, for different ability levels $\alpha_i\in[50\%,100\%]$.  Here, we have used the parameters $m:=80\%$ and $\hat{\alpha}:=60\%$.  The higher-performing students will see more muted reductions of their study time in equilibrium.}
\end{center}
\end{figure}

\begin{figure}[t]
\begin{center}
\includegraphics[height=200px]{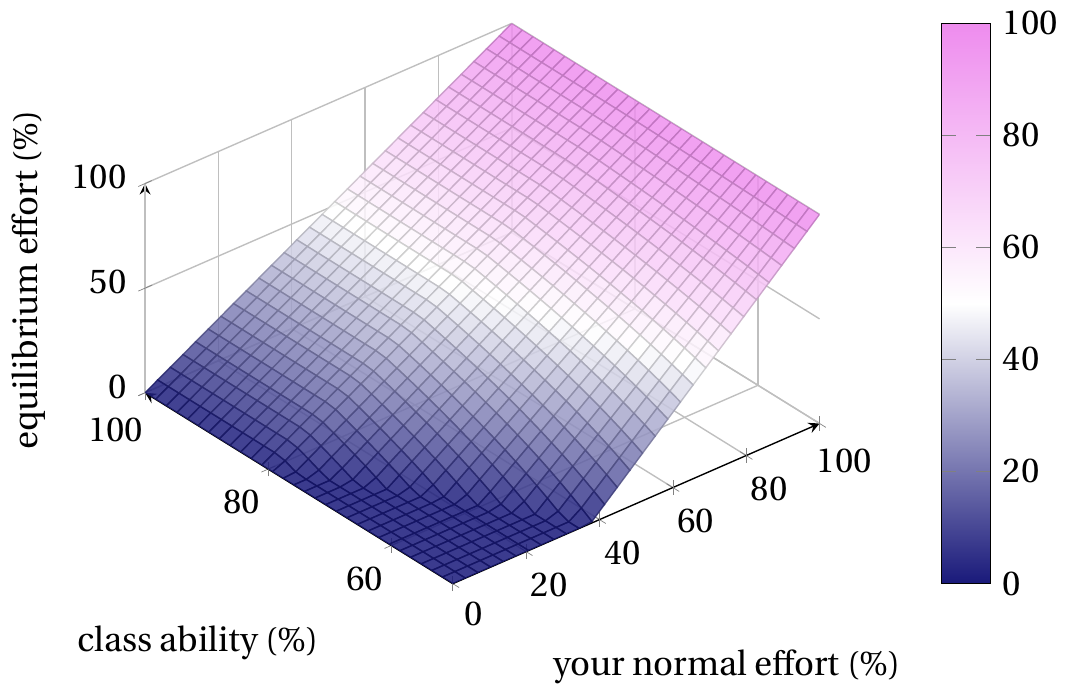}
\caption{\sc Equilibrium strategy card for Garivaltis students ($m:=80\%$, $n:=\infty$).  The lower your ``preference'' for leisure ($=1-\alpha_i$), the harder you should try.  You should study more if you believe that the ability $\hat{\alpha}$ of the class is high.  In response to the disincentives of a curve, you should increase your total leisure hours by the same \textit{\textbf{percentage}} that all your classmates increase theirs.  Good luck.}
\end{center}
\end{figure}

\begin{lemma}[Positive Effort in the Make-or-Break Region]\label{poseffort}
There is no point $x_{-i}$ in player $i$'s make-or-break region that has the following two properties simultaneously:
\begin{itemize}
\item
The non-zero critical points $C(x_{-i})\backslash\{0\}$ both break the curve\footnote{For values of $x_{-i}$ in student $i$'s make-or-break region, $x_i=0$ will always make the curve.  Thus, in this happenstance, the relevant choice for the pupil is to decide whether to break the curve with full effort, or to make the curve with zero effort.  In this particular situation, the critical point $\alpha_i-(1-\alpha_i)\left(nm/(n-1)-\overline{x}_{-i}\right)$ is inferior to $\alpha_i$, since they both break the curve.  On the other hand, if $\alpha_i-(1-\alpha_i)\left(nm/(n-1)-\overline{x}_{-i}\right)$ makes the curve, then $x_i=0$ becomes the irrelevant critical point;  this is a separate case that will be dealt with below.};
\item
$U_i(0,x_{-i})\ge U_i(\alpha_i,x_{-i})$, e.g., zero effort for agent $i$ gives higher utility than full effort.
\end{itemize}Consequently, we have $0\notin\argmax\limits_{x_i\in C(x_{-i})}U_i(x_i,x_{-i})$;  which is to say, if $x_{-i}$ lies in player $i$'s make-or-break region, then zero is not a best response to $x_{-i}$, and we have 
\begin{equation}
\argmax\limits_{x_i\in C(x_{-i})}U_i(x_i,x_{-i})\subseteq\left\{\underbrace{\alpha_i-(1-\alpha)\left(\frac{nm}{n-1}-\overline{x}_{-i}\right)}_{\text{left-hand critical point, }x_i^{(L)}},\alpha_i\right\}.
\end{equation}
\end{lemma}

\begin{proof}
Assume that the stated conditions on $x_{-i}$ are all in effect;  we proceed to derive a contradiction, i.e., the intersection of all these constraints on $x_{-i}$ is the empty set.  The non-zero critical points $x_i\in C(x_{-i})\backslash\{0\}$ will both break the curve if and only if the leftmost\footnote{The factor $\left(nm/\left(n-1\right)-\overline{x}_{-i}\right)$ in the expression for $x^{(L)}_i$ is non-negative when $\overline{x}_{-i}$ lies in the make-or-break-region.} point $x_i^{(L)}:=\alpha_i-(1-\alpha_i)\left(nm/\left(n-1\right)-\overline{x}_{-i}\right)$ does, because if some value of $x_i$ breaks the curve, then all higher values do as well.  Now, according to Proposition \ref{complement}, such breakage is characterized by an effort level $x_i$ that is higher than $\hat{x}_i=nm-(n-1)\overline{x}_{-i}$.  Thus, solving the inequality
\begin{equation}
\alpha_i-\left(1-\alpha_i\right)\left(\frac{nm}{n-1}-\overline{x}_{-i}\right)\ge nm-(n-1)\overline{x}_{-i}
\end{equation}for $\overline{x}_{-i}$, we get the simplified condition
\begin{equation}\label{simpcondition}
\overline{x}_{-i}\geq\frac{nm}{n-1}-\frac{\alpha_i}{n-\alpha_i}
\end{equation}that expresses our first bullet point above.  Now, the second bullet point says that
\begin{equation}
\left(m-\frac{n-1}{n}\overline{x}_{-i}\right)^{\alpha_i}\ge \alpha_i^{\alpha_i}(1-\alpha_i)^{1-\alpha_i},
\end{equation}which, when solved for $\overline{x}_{-i}$, means that
\begin{equation}\label{meansthat}
\overline{x}_{-i}\leq\frac{n}{n-1}\left(m-\alpha_i(1-\alpha_i)^{\frac{1-\alpha_i}{\alpha_i}}\right).
\end{equation}If it is possible to have the inequalities (\ref{simpcondition}) and (\ref{meansthat}) hold simultaneously, then the model parameters $\theta:=\left(\alpha_1,...,\alpha_n,m\right)\in\Theta$ must satisfy the condition
\begin{equation}\label{proved}
(n-\alpha_i)(1-\alpha_i)^{\frac{1-\alpha_i}{\alpha_i}}\leq\frac{n-1}{n},
\end{equation}which is impossible, as we will show presently.  To that end, let $f_n(z):=(n-z)\left(1-z\right)^{\frac{1-z}{z}}$.  We will demonstrate that $\lim\limits_{z\to0}f_n(z)=n/e>1-1/n$, and that $f_n(\bullet)$ is an increasing function over the interval $[0,1]$, whence the inequality (\ref{proved}) will be false for all possible parameter values.  Note that $\lim\limits_{z\to0}(1-z)^{\frac{1-z}{z}}=1/e$, since, taking logs, we have
\begin{equation}
\lim\limits_{z\to0}\left(\frac{\log(1-z)}{z}-\log(1-z)\right)=\lim\limits_{z\to0}\frac{-1}{1-z}=-1,
\end{equation}by L'H\^{o}pital's rule.  Now, the inequality $n/e>1-1/n$ is equivalent to saying that $e<n+1+\frac{1}{n-1}$, which is true for all $n\geq2$.
\par
As to the fact that $f_n(\bullet)$ is increasing, we reckon that
\begin{equation}\label{reckon}
\frac{d}{dz}\left[\log f_n(z)\right]=-\frac{1}{z}\left(\frac{n}{n-z}+\frac{\log(1-z)}{z}\right).
\end{equation}Now, let $g_n(z):=nz+(n-z)\log(1-z)$.  We will show that $g_n(z)\leq0$ for all $z\in[0,1)$.  Note that $g_n(0)=0$, and that $g_n(\bullet)$ is decreasing over the interval $[0,1)$.  The reason for this is as follows.  After differentiating $g_n(\bullet)$, we obtain
\begin{equation}
g_n'(z)=-(n-1)\frac{z}{1-z}+\log\left(\frac{1}{1-z}\right).
\end{equation}Thus, in order to prove that $g_n'(z)\le0$ for all $n\geq2$, it suffices to demonstrate the truth of the relation for $n=2$, since $g_n'(z)$ is decreasing in $n$ for all $z\in[0,1)$.  After a bit of re-arranging, the statement that $g_n'(z)\le0$ is equivalent to the assertion that $e^{z/(1-z)}\ge1/(1-z)$.  Finally, then, in the bound $e^{\zeta}\geq1+\zeta$, which holds good over the whole real axis\footnote{The graph of the convex function $\zeta\mapsto e^{\zeta}$ must lie above all of its tangents, in particular, the line $\zeta\mapsto1+\zeta$.}, we put $\zeta:=z/(1-z)$ in order to obtain the \textit{coup de gr\^{a}ce}
\begin{equation}
\exp\left(\frac{z}{1-z}\right)\ge1+\frac{z}{1-z}=\frac{1}{1-z},
\end{equation}and our Lemma is hereby established.
\end{proof}

\begin{remark}
If the non-zero critical points $C(x_i)\backslash\{0\}$ do not both break the curve, then the left-hand point $x_i^{(L)}=\alpha_i-(1-\alpha_i)\left(nm/(n-1)-\overline{x}_{-i}\right)$ is guaranteed to make the curve;  for, if $x_i=\alpha_i$ makes the curve, then so too does every smaller value of $x_i$.  In that case, $x_i=0$ is not the best response, since $x_i^{(L)}$ uniquely satisfies the first order condition of the log-concave program (\ref{objective}).  Thus Lemma \ref{poseffort} rules out the only possible situation where we could have had a zero-effort best response over the make-or-break region, and we are free to focus our attention on the remaining pair of critical points.  
\end{remark}

On the strength of the all the foregoing theory, we are at last in a position to give a definitive, final formula for each player's best response correspondence, a formula which is rich in its consequences.

\begin{theorem}[Pinpointing the Jump in the Reaction Correspondence]\label{jumptheorem}
The function  
\begin{equation}\label{pinpoint}\boxed{
\phi_i(z):=\left(m+\frac{n-1}{n}\left(1-z\right)\right)^{\alpha_i}\left(1+\frac{nm}{n-1}-z\right)^{1-\alpha_i}-1}
\end{equation}has a unique zero, $J_i$, in the interval
\begin{equation}\label{interval}
\left[\frac{nm-\alpha_i}{n-1},\frac{nm}{n-1}-\frac{\alpha_i}{n-\alpha_i}\right],
\end{equation}which is a subset of player $i$'s make-or-break region.  If student $i$'s opponents play an action profile $x_{-i}$ whose sample mean is $J_i$ (viz., $\overline{x}_{-i}=J_i=\phi_i^{-1}(0)$), then agent $i$ is exactly indifferent between the pair of non-zero critical points $C(x_{-i})\backslash\{0\}=\left\{x_i^{(L)},\alpha_i\right\}$, meaning that $U_i\left(x_i^{(L)},x_{-i}\right)=U_i\left(\alpha_i,x_{-i}\right)$.  Against such values\footnote{In general, there is a continuum (with $n-2$ degrees of freedom) of opponent action profiles $x_{-i}$ that have a sample mean equal to $J_i$.  This generates a tear, or bifurcation, in player $i$'s best response correspondence.  In the two-person game, whereby the sample mean $\overline{x}_{-i}$ is just $x_{-i}$ itself, $x_{-i}=J_i$ will be the lone point at which there is a jump in the graph of student $i$'s best response correspondence.} of $x_{-i}$, the left-hand critical point $x_i^{(L)}$ always makes the curve, and the right-hand critical point $\alpha_i$ always breaks the curve.  For $\overline{x}_{-i}<J_i$ the unique best response is $x_i^{(L)}$, which makes the curve, and for $\overline{x}_{i}>J_i$ the unique best response is $\alpha_i$, which breaks the curve.  Thus, the complete formula for player $i$'s best response correspondence is given by

\[\boxed{ BR_i\left(x_{-i}\right)=\begin{cases}
      \alpha_i & \text{ if  }J_i<\overline{x}_{-i}\leq1\text{   (no-curve region)}\\

\left\{\alpha_i-(1-\alpha_i)\left(\frac{nm}{n-1}-J_i\right),\alpha_i\right\}& \text{ if }\overline{x}_{-i}=J_i\text{  (indifference hyperplane)}\\
      
      \alpha_i-\left(1-\alpha_i\right)\left(\frac{nm}{n-1}-\overline{x}_{-i}\right) & \text{ if } \frac{nm}{n-1}-\frac{\alpha_i}{1-\alpha_i}\leq\overline{x}_{-i}<J_i\text{  (curve region)}\\
      
0 & \text{ if }0\leq\overline{x}_{-i}\leq\frac{nm}{n-1}-\frac{\alpha_i}{1-\alpha_i}\text{ (no-show region).}      
   \end{cases}
}\]

\end{theorem}

\begin{proof}
To start, the interval (\ref{interval}) is clearly a subset of the make-or-break interval (\ref{complement}).  Now, $\phi_i(z)$ is a strictly decreasing function of $z$, so that it can have at most one root in the interval (\ref{interval}).  We proceed to show that $\phi_i\left(\left(nm-\alpha_i\right)/\left(n-1\right)\right)\ge1$ and $\phi_i\left(nm/\left(n-1\right)-\alpha_i/\left(n-\alpha_i\right)\right)\le1$, whence the existence of the root will obtain from the intermediate value theorem.  Thus, there lies
\begin{equation}
\phi_i\left(\frac{nm-\alpha_i}{n-1}\right)=\left(\frac{\alpha_i+n-1}{n}\right)^{\alpha_i}\left(1+\frac{\alpha_i}{n-1}\right)^{1-\alpha_i}-1=\frac{\alpha_i+n-1}{n^{\alpha_i}(n-1)^{1-\alpha_i}}-1.
\end{equation}Hence, we need to show that $n^{\alpha_i}(n-1)^{1-\alpha_i}\leq\alpha_i+n-1$;  but this is an immediate consequence of the AGM inequality applied to the geometric mean $n^{\alpha_i}(n-1)^{1-\alpha_i}$.  As to the other endpoint, we have
\begin{multline}
\phi_i\left(\frac{nm}{n-1}-\frac{\alpha_i}{n-\alpha_i}\right)=\left(\frac{n-1}{n}\left(1+\frac{\alpha_i}{n-\alpha_i}\right)\right)^{\alpha_i}\left(1+\frac{\alpha_i}{n-\alpha_i}\right)^{1-\alpha_i}-1\\
=\frac{(n-1)^{\alpha_i}n^{1-\alpha_i}}{n-\alpha_i}-1.
\end{multline}Thus, we require the fact that $(n-1)^{\alpha_i}n^{1-\alpha_i}\leq n-\alpha_i$, which follows at once by applying the AGM inequality to the geometric mean $(n-1)^{\alpha_i}n^{1-\alpha_i}$.  Hence, we have demonstrated that $\phi_i(\bullet)$ has a unique root in the interval (\ref{interval}).
\par
Now, we will address all four segments of the make-or-break region, moving from left to right.  We begin with the segment
\begin{equation}\label{seg1}
\overline{x}_{-i}\in\left[\frac{nm-1}{n-1},\frac{nm-\alpha_i}{n-1}\right].
\end{equation}For these values of $\overline{x_i}$, we have $\hat{x}_i\geq\alpha_i$ so that all the critical points of $i$'s utility $U_i\left(\bullet,x_{-i}\right)$ will make the curve.  Thus, over the segment (\ref{seg1}), agent $i$'s best response is to play $x_i^{(L)}$.  Next, we have the segment 
\begin{equation}\label{seg2}
\overline{x}_{-i}\in\left(\frac{nm-\alpha_i}{n-1},J_i\right).
\end{equation}For these values of $\overline{x}_{-i}$, we have the relations $x_i^{(L)}<\hat{x}_i<\alpha_i$.  Note that the inequality $x_i^{(L)}<\hat{x}_i$ is true for all $\overline{x}_{-i}$ such that $\overline{x}_{-i}<\frac{nm-\alpha_i}{n-1}$;  here, we have $\overline{x}_{-i}<J_i\leq\frac{nm-\alpha_i}{n-1}$.  Thus, over the segment (\ref{seg2}), $\alpha_i$ breaks the curve and $x_i^{(L)}$ makes the curve.  In this happenstance, we need to show that $U_i\left(x_i^{(L)},x_{-i}\right)>U_i\left(\alpha_i,x_{-i}\right)$.  Simplifying the ratio of these two utilities, we get the expression
\begin{equation}
\frac{U_i\left(x_i^{(L)},x_{-i}\right)}{U_i\left(\alpha_i,x_{-i}\right)}-1=\left(m+\frac{nm}{n-1}\left(1-\overline{x}_{-i}\right)\right)^{\alpha_i}\left(1+\frac{nm}{n-1}-\overline{x}_{-i}\right)^{1-\alpha_i}-1=\phi_i\left(\overline{x}_{-i}\right).
\end{equation}Hence, our goal is simply to show that $\phi_i\left(\overline{x}_{-i}\right)>0$.  This follows immediately from the fact that $\phi_i(\bullet)$ is strictly decreasing:  in the hypothesis $\overline{x}_{-i}<J_i$, we apply $\phi_i(\bullet)$ to both sides, and obtain $\phi_i\left(\overline{x}_{-i}\right)>\phi_i\left(J_i\right)=0$, by the definition of $J_i$.  Therefore, we have $BR_i\left(x_{-i}\right)=x_i^{(L)}$, as promised.
\par
Next, if $\overline{x}_{-i}=J_i$, then we clearly have $U_i\left(x_i^{(L)},x_{-i}\right)=U_i\left(\alpha_i,x_{-i}\right)$ (since $\phi_i\left(\overline{x}_{-i}\right)=0$), so that $BR_i\left(x_{-i}\right)=\left\{x_i^{(L)},\alpha_i\right\}$.  Again, this value of $\overline{x}_{-i}$ is such that $\alpha_i$ breaks the curve and $x_i^{(L)}$ makes the curve.  Moving on, we proceed with the segment
\begin{equation}\label{seg3}
\overline{x}_{-i}\in\left(J_i,\frac{nm}{n-1}-\frac{\alpha_i}{n-\alpha_i}\right].
\end{equation}Over this set of non-$i$ action profiles, we have $\phi_i\left(\overline{x}_{-i}\right)<0$, so that $BR_i\left(\overline{x}_{-i}\right)=\alpha_i$.  That is, the situation in the interval (\ref{seg3}) is that $\alpha_i$ breaks the curve and $x_i^{(L)}$ makes the curve, but we have $U_i\left(\alpha_i,x_{-i}\right)>U_i\left(x_i^{(L)},x_{-i}\right)$.  To finish the proof, we must dispose of the interval
\begin{equation}
\overline{x}_{-i}\in\left(\frac{nm}{n-1}-\frac{\alpha_i}{n-\alpha_i},\frac{nm}{n-1}\right].
\end{equation}For such values of $\overline{x}_{-i}$, all the non-zero critical points break the curve, and $x_i=0$ makes the curve.  According to Lemma \ref{poseffort}, zero effort can never be a best response in this situation.  Thus, since $x_i^{(L)}$ and $\alpha_i$ both break the curve, and $0\notin BR(x_{-i})$, student $i$'s best play is to break the curve, and we have $BR_i(x_{-i})=\alpha_i$, which completes the proof of the Theorem.
\end{proof}

\begin{figure}[t]
\begin{center}
\includegraphics[height=200px]{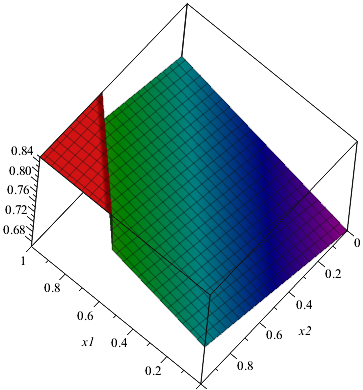}
\caption{\sc Illustration of player 3's best response surface in the 3-person game.  This example uses the parameter values $\left(\alpha_3,m\right):=\left(85\%,80\%\right)$.  The bifurcation, or jump, in the graph of $BR_3(x_1,x_2)$ occurs over the line $x_1+x_2=2J_3$.  Here, we have $J_3=\phi_3^{-1}(0)=78.9\%$.}
\end{center}
\end{figure}

\begin{figure}[t]
\begin{center}
\includegraphics[height=200px]{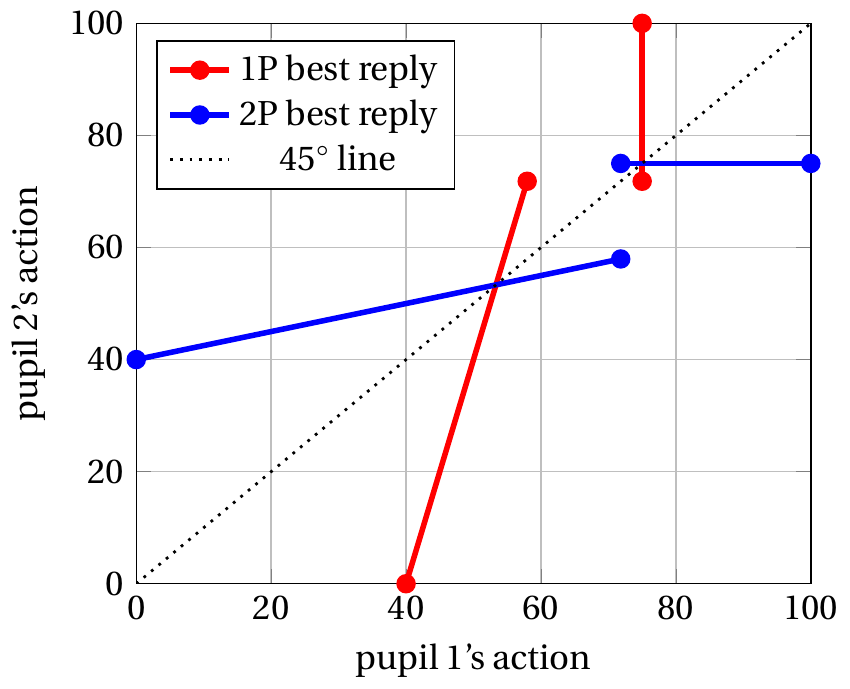}
\caption{A typical example of multiple equilibria in the 2-person game.  Here, we used the parameters $m:=70\%$ and $\alpha_1=\alpha_2=75\%$, so that the reaction correspondences are pleasantly symmetric with respect to the $45^\circ$ line.  Since $n:=2$, the opposing sample mean $\overline{x}_{-i}$ is equal to $x_{-i}$ itself.  The jump in each player's best reply occurs at $\left(x_i,x_{-i}\right)=\left(0.58,0.718\right)$.  We have the no-curve (``try-hard'') equilibrium $y^*=\left(0.75,0.75\right)$, and the curved interior equilibrium $x^*=\left(0.533,0.533\right)$.  Note well the lack of oddness (due to the non-convexities).  Both players receive lower grades under $x^*$ ($G_i=70\%$) than they do under $y^*$ ($G_i=75\%$), but the low-effort equilibrium (always) Pareto dominates the high-effort equilibrium;  we have $U_i(x^*)\equiv0.633$ and $U_i(y^*)\equiv0.57$ for $i=1,2.$}
\end{center}
\end{figure}

Note well that, although each player's best response correspondence has a closed graph, it has a non-convexity due to its jump behavior, which obtains from a lack of quasi-concavity in player $i$'s own action.  That is, at jump points $\overline{x}_{-i}=J_i$, player $i$ has a pair of distant best responses (one high effort, one low effort), and the points in between are all missing from the graph of $BR_i(\bullet)$.  Thus, the Kakutani fixed point theorem (\cite{kakutani}) does not apply;  however, the monotonic character of these best responses (which, as noted above, is a general implication of the \cite{topkispaper} theory) means that the general existence of equilibria in our university model is instead due to the Knaster-Tarski fixed point theorem (cf. with \cite{knaster,tarsky,mwg}).

Having fully elaborated the exact behavior of each player's reaction correspondence, we proceed to strengthen Corollary \ref{basic} and give conditions on the parameter vector $\theta$ that are both necessary and sufficient for the existence of a full effort (uncurved) equilibrium.  For parameter vectors $\theta=\left(\alpha_1,...,\alpha_n,m\right)\in\Theta$ that generate a no-curve equilibrium, the corresponding strategy profile $y^*=\left(\alpha_1,...,\alpha_n\right)$ will constitute the greatest equilibrium point with respect to the vector partial order $\leq$. 

\begin{theorem}[Parameter Set for the Try-Hard Equilibrium]
There exists a no-curve equilibrium (in which each kid $i$ exerts his or her fullest undominated effort $y^*_i=\alpha_i$) if and only if the parameters of the model satisfy the condition
\begin{equation}\label{necessary}\boxed{
\overline{\alpha}\ge\frac{1}{n}\max\limits_{1\leq i\leq n}\left[\left(n-1\right)J_i(\alpha_i,n,m)+\alpha_i\right].}
\end{equation}The right-hand-side of this inequality converges to the instructor's target mean, $m$, as the class size $n$ tends to infinity.
\end{theorem}
\begin{proof}
Assume that there exists an uncurved equilibrium $y^*=\left(y^*_1,...,y^*_n\right)$.  Then, we must have $y^*_i=\alpha_i$ for all $i$;  as we have noted above, there can be situations where the point $\alpha_i-\left(1-\alpha_i\right)\left[nm/\left(n-1\right)-J_i\right]$ ties with $\alpha_i$, but in that happenstance, choosing the lower effort level will break the curve, contradicting our hypothesis.  Thus, we must have $y^*_i=\alpha_i$ for all $i$.  Now, reading off from our exact formula for $BR_i(\bullet)$, the only way that $\alpha_i$ can be a best response to $\alpha_{-i}$ (or a best response to anything at all) is when $\overline{\alpha}_{-i}\geq J_i$.  Thus, since $\alpha_i$ is assumed to be a best response to $\alpha_{-i}$ for all $i=1,...,n$, we must have
\begin{equation}\label{numbers}
n\overline{\alpha}\ge\left(n-1\right)J_i+\alpha_i
\end{equation}for all players $i$.  Hence, we obtain the necessary condition (\ref{necessary}).  Conversely, suppose that the inequality (\ref{necessary}) holds true.  Then re-arranging (\ref{numbers}), we have the fact that $\overline{\alpha}_{-i}\geq J_i$ for all $i$, whence $\alpha_i\in BR_i\left(\alpha_{-i}\right)$ for every player $i$, so that $\left(\alpha_1,...,\alpha_n\right)$ is an equilibrium point of $\Gamma_n$.
\par
As to the limiting behavior of the characterization (\ref{necessary}), we have
\begin{equation}\label{sequence}
\frac{1}{n}\max\limits_{1\leq i\leq n}\left[\left(n-1\right)J_i+\alpha_i\right]=\frac{n-1}{n}J_{i^*(n)}+\frac{\alpha_{i^*(n)}}{n}
\end{equation}for some index $i^*(n)$.  Since $J_{i^*(n)}$ always lies in the make-or-break region, and the endpoints of that interval tend to $m$ as $n\to\infty$, we have, by the squeezing process (cf. with \cite{lang}), $\lim\limits_{n\to\infty}J_{i^*(n)}=m$.  Similarly, we have $0\leq\alpha_{i^*(n)}/n\leq1/n$, so that $\lim\limits_{n\to\infty}\alpha_{i^*(n)}/n=0$.  Thus, the sequence (\ref{sequence}) converges to the professor's target mean as the class size becomes infinite.  Q.E.D.
\end{proof}
In a similar vein, having exactly pinpointed the location of the jumps (\ref{pinpoint}), we can read off from the best response correspondence a necessary and sufficient condition on model parameters in order that the interior, curved equilibrium (\ref{equilibrium}) obtains.  In words, the total number of free points that the professor gives away must be sufficiently large in order to guarantee that the equilibrium is curved, but sufficiently small in order to guarantee that the worst  student actually shows up to the exam.
\begin{theorem}[Parameter Set for the Curved Interior Equilibrium]
The parameter vector $\theta:=\left(\alpha_1,...,\alpha_n,m\right)\in\Theta$ generates the interior, curved equilibrium point (\ref{equilibrium}) if and only if the following condition is satisfied:
\begin{equation}\label{0dc}\boxed{
\max_{1\leq i\leq n}\left\{\left(n-\alpha_i\right)\left(\frac{nm}{n-1}-J_i\right)-\alpha_i\right\}\le n\left(m-\overline{x}^*\right)<\frac{\left(n-1\right)\alpha_{(1)}}{1-\alpha_{(1)}},}
\end{equation}where $\overline{x}^*:=1-nm\left(\hat{\alpha}_n^{-1}-1\right)/(n-1)$, $\hat{\alpha}_n:=n-\text{HarmonicMean}\left(n-\alpha_1,...,n-\alpha_n\right)$, $J_i\left(\alpha_i,n,m\right)=\phi_i^{-1}(0)$ is the location of the jump in student $i$'s best response correspondence, and $\alpha_{(1)}:=\min\limits_{1\leq i\leq n}\alpha_i$ is the first order statistic (the minimum) of the students' ability parameters.
\end{theorem}
\begin{proof}
Based on the final, exact formula for each player's reaction correspondence that was given in Theorem \ref{jumptheorem}, the curved equilibrium play $x^*=\left(x_i^*\right)_{i=1}^n$ that was specified in (\ref{equilibrium}) obtains if and only if $x^*$ has the property that, for each player $i$, the sample mean $\overline{x^*}_{-i}$ of his or her opponents' effort is less than or equal to $J_i$.  Since the formula (\ref{equilibrium}) for $x_i^*$ is strictly increasing in $\overline{x}_{-i}$, the condition that $\overline{x^*}_{-i}\leq J_i$ is equivalent to the statement that
\begin{equation}\label{statemt}
x_i^*\leq\alpha_i-(1-\alpha_i)\left(\frac{nm}{n-1}-J_i\right).
\end{equation}Thus, combining (\ref{equilibrium}) with (\ref{statemt}), and solving the resulting inequality for $\overline{x}^*$, we obtain the relation
\begin{equation}
\overline{x}^*\leq m+\frac{1}{n}\min_{1\leq i\leq n}\left\{\alpha_i-\left(n-\alpha_i\right)\left(\frac{nm}{n-1}-J_i\right)\right\},
\end{equation}as promised.  Finally, in order to guarantee the interiority of the curved equilibrium, we must ensure that $\overline{x}^*_{-i}$ does not lie in the no-show region of any player $i$;  this is equivalent to saying that $x_i^*>0$, where $x_i^*$ is given by (\ref{equilibrium}).  Thus, we must have $n\left(\overline{x}^*-m\right)>-(n-1)\alpha_i/\left(1-\alpha_i\right)$ for all $i$, which is equivalent to the stated relation
\begin{equation}
n\left(\overline{x}^*-m\right)>-(n-1)\left(\frac{1}{1-\min\limits_{1\leq i\leq n}\alpha_i}-1\right).
\end{equation}Q.E.D.
\end{proof}

\begin{figure}[t]
\begin{center}
\includegraphics[height=200px]{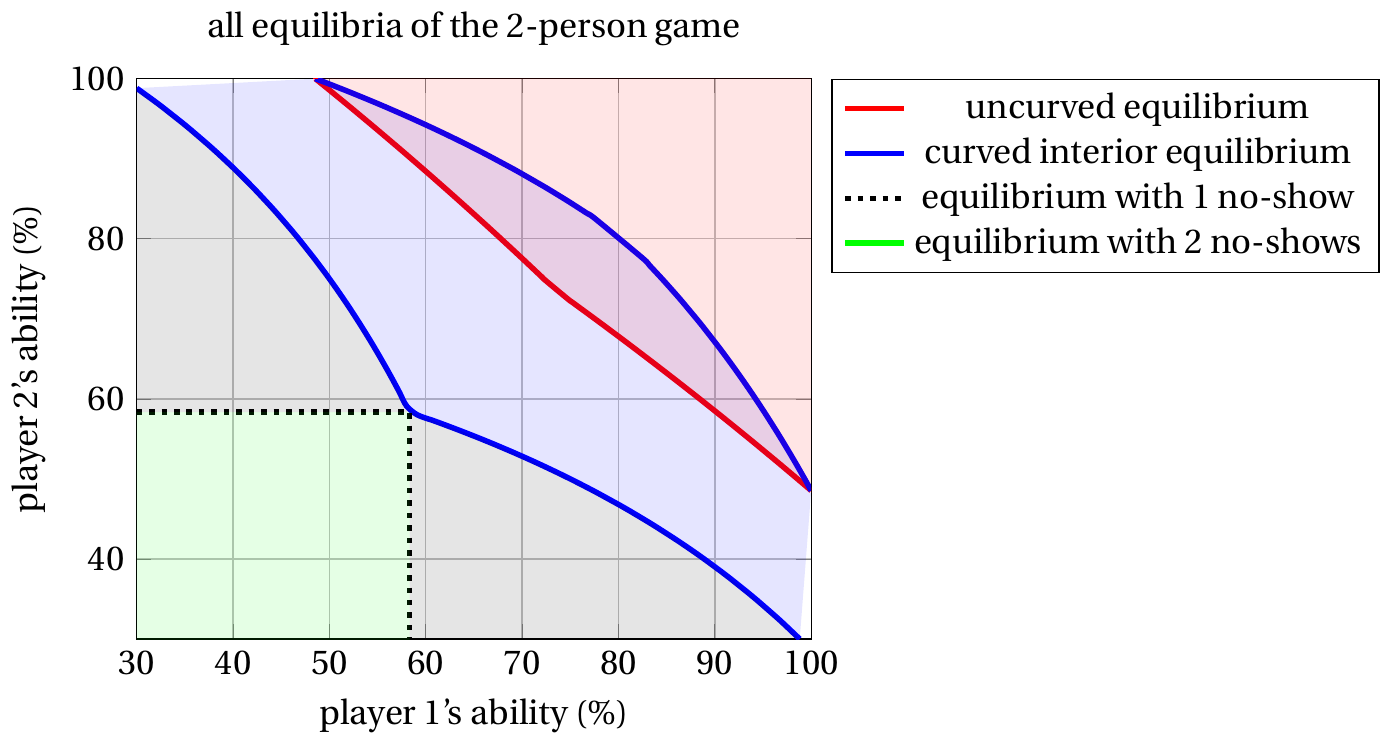}
\caption{\sc Bivariate cross section ($m:=70\%$) of the parameter space $\Theta$ for the 2-person game.  For ability vectors $\alpha:=\left(\alpha_1,\alpha_2\right)$ in the red region, a no-curve (``try-hard'') equilibrium exists;  in the blue region, we have a curved (interior) equilibrium.  In the black region, we have an equilibrium whereby the weaker student is a no-show;  in the green region, we have a 2-don't care equilibrium, with both players receiving a curved grade of $70\%$.  Note well that in the intersection of the red and blue regions, we have an even number of equilibria (viz., 2).}
\end{center}
\end{figure}

Having characterized the exact subsets of the parameter space $\Theta$ that support uncurved (``try-hard'') and curved interior (all-try) equilibria, we proceed to give formulas for all remaining types of equilibrium points, which we will refer to as ``$k$-don't care'' equilibria.  This means that the bottom $k$ students are no-shows\footnote{Such an equilibrium has obtained for very small numbers of no-shows ($k=1$ or $k=2$ didn't care) a couple of times in the lived experience of the author.  Of course, although the model will generally have multiple equilibria (depending on the parameter vector, $\theta$), equilibrium selection is up to the students;  the precise equilibrium play that obtains in the university will ultimately depend on the extent to which the agents succeed in coordinating their respective effort levels, for mutual benefit.} in the equilibrium profile $x^{*(k)}=\left(x_1^{*(k)},...,x_n^{*(k)}\right)$, i.e., they give zero effort ($x_i^{*(k)}=0$).

\begin{theorem}[Formulas for $k$-Don't Care Equilibria]\label{dontcare}
In any equilibrium strategy profile $x^*=\left(x^*_1,...,x^*_n\right)$, if some player $i$ is a no-show (meaning that $x^*_i=0$), then all weaker students (i.e., all students with lower ability parameters) will also be no-shows;  if a given player shows up to the exam ($x^*_i>0$), then all stronger students will also show up to the exam in equilibrium.  In a Nash equilibrium $x^{*(k)}$ whereby the bottom $k$ students don't care, the efforts of the top $n-k$ students will be given by the formulas
\begin{equation}\label{kdcformula}\boxed{
x_i^{*(k)}=\frac{(n-1)\alpha_i-n(1-\alpha_i)\left(m-\overline{x}^{*(k)}\right)}{n-\alpha_i},}
\end{equation}where the $k$-don't care mean $\overline{x}^{*(k)}$ consists in the expression
\begin{equation}\boxed{
\overline{x}^{*(k)}=\frac{(n-1)(m+1)S_2-\left(n-k\right)\left(m+1-1/n\right)}{(n-1)\left(S_2-1\right)+k},}
\end{equation}where 
\begin{equation}\label{sum}\boxed{
S_2:=\sum\limits_{i=k+1}^n\frac{1}{n-\alpha_{(i)}},}
\end{equation}and $\alpha_{(i)}$ denotes the $i^{th}$ order statistic of the ability vector $(\alpha_1,...,\alpha_n)$, i.e., the sum (\ref{sum}) is taken over the ability parameters of the top $n-k$ students in the class.
\end{theorem}
\begin{remark}
The $n$-don't care mean (whereby there are $n$ no-shows, and nobody cares) is $\overline{x}^{*(n)}=0$, and we have $x_i^{*(n)}\equiv0$ for all $i$.  Note that for $k=n$, we get $S_2=0$ (the empty sum), so that the numerator of (\ref{kdcformula}) is equal to zero.  On the other extreme, if we put $k=0$ (zero don't care), then we get the curved interior equilibrium (\ref{equilibrium}).
\end{remark}

\begin{proof}
First, assume that player $i$ is a no-show in the Nash equilibrium profile $x=\left(x_1,...,x_n\right)$.  Then, $\overline{x}_{-i}$ must lie in player $i$'s no-show region, viz., $\overline{x}_{-i}\leq nm/\left(n-1\right)-\alpha_i/\left(1-\alpha_i\right)$.  Now, consider some other player $j\neq i$ that has a lower ability parameter, $\alpha_j\le\alpha_1$.  Then, since $x_i=0$, we have the inequalities
\begin{equation}
\overline{x}_{-j}\leq\overline{x}_{-j}+\frac{x_j}{n-1}=\overline{x}_{-i}\le\frac{nm}{n-1}-\frac{\alpha_i}{1-\alpha_i}\leq\frac{nm}{n-1}-\frac{\alpha_j}{1-\alpha_j},
\end{equation}so that $\overline{x}_{-j}$ belongs to student $j$'s no-show region.  Thus, since pupil $j$ plays a best response to $x_{-j}$ in the equilibrium profile $x=\left(x_j,x_{-j}\right)$, we must have $x_j=0$, so that agent $j$ is also a no-show.
\par
Next, retracing our steps in the proof of Theorem \ref{nstudent}, in general we must sum equation (\ref{summing}) over all the students who show up to exam, viz., the students whose ability parameters are $\left(\alpha_{(k+1)},...,\alpha_{(n)}\right)$, assuming that the equilibrium profile $x$ contains exactly $k$ no-shows.  Thus, the auxiliary sums (\ref{auxsums}) turn into
\begin{equation}\label{generalized}
S_1:=\sum_{i=k+1}^n\frac{\alpha_{(i)}}{n-\alpha_{(i)}}\text{ and }S_2:=\sum_{i=k+1}^n\frac{1}{n-\alpha_{(i)}},
\end{equation}where we have the relation $S_1=n(S_2-1)+k$.  Then, with the generalized expressions (\ref{generalized}) in hand, the expression (\ref{sumexpr}) for the aggregate equilibrium effort gives us
\begin{equation}
\overline{x}^{*(k)}=\frac{(n-1)(m+1)S_2-(n-k)(m+1-1/n)}{(n-1)(S_2-1)+k},
\end{equation}which completes the proof.
\end{proof}

\begin{figure}[t]
\begin{center}
\includegraphics[height=200px]{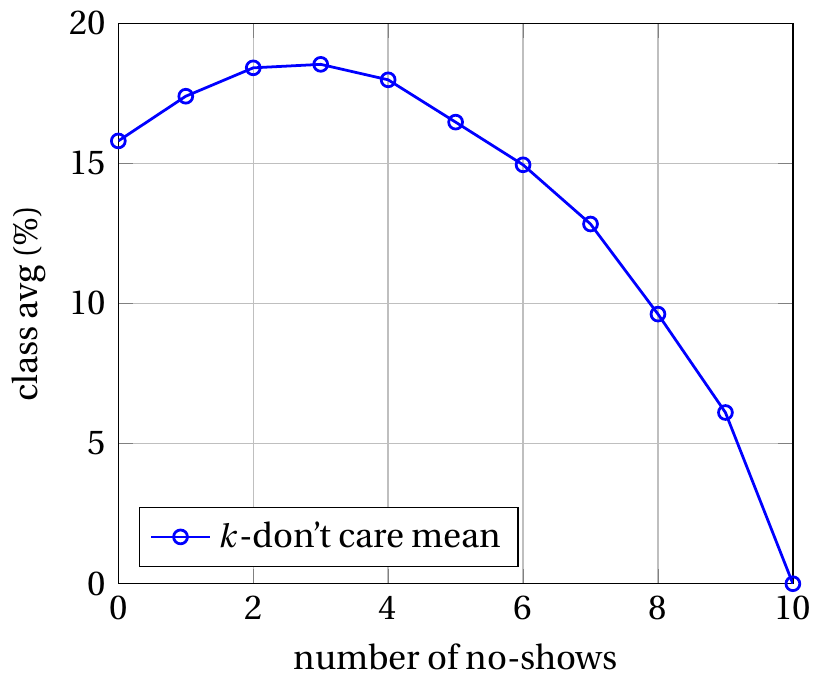}
\caption{\sc The $k$-don't care means $\overline{x}^{*(k)}$ for different values of $k$ in a 10-student course.  Here, we have used the parameter values $m:=85\%$ and $\alpha:=\left(.38,.39,.42,.45,.5,.51,.55,.62,.65,.8\right)$.  In general, the $k$-don't care means $\overline{x}^{*(k)}$ are not guaranteed to decrease monotonically in $k$, although such will often be the case in practical examples.  However, if a $k$-don't care equilibrium and an $l$-don't care equilibrium co-exist in the same model, with $k\le l$, then we always have $\overline{x}^{*(k)}\ge\overline{x}^{*(l)}$.  Be advised that this is a highly exaggerated (counter-) example;  the average ability is just $52.7\%$, and the curve is outlandishly generous.}

\end{center}
\end{figure}

Given the insights that we have just derived in Theorem \ref{dontcare} above, it becomes a simple matter to identify the exact parameter set that supports an equilibrium with $k$ no-shows;  thus, with the advent of the following Corollary, we will have fully characterized and classified all possible types of university equilibria that can occur in our model.

\begin{corollary}[Parameter Sets for $k$-Don't Care Equilibria]
The model's parameter vector $\theta\in\Theta$ supports an equilibrium with exactly $k$ no-shows (by the bottom $k$ students in the course) if and only if
\begin{equation}\label{kdcinterval}\boxed{
\alpha_{(k)}\leq1-\left(1+\frac{n}{n-1}\left(m-\overline{x}^{*(k)}\right)\right)^{-1}<\alpha_{(k+1)},}
\end{equation}where $\overline{x}^{*(k)}$ is the $k$-don't care mean, and $\alpha_{(k)},\alpha_{(k+1)}$ are the $k^{th}$ and $\left(k+1\right)^{st}$ order statistics of the ability vector, respectively.  A nobody-cares equilibrium ($k=n$) will exist if and only if $\alpha_{(n)}\le m/\left(m+1-1/n\right)$.
\begin{proof}
In the formula (\ref{kdcformula}), we require that the numerator is strictly positive for the top $n-k$ students in the course, viz., those whose ability parameters are greater than or equal to $\alpha_{(k+1)}/\left(1-\alpha_{(k+1)}\right)$.  Thus, we must have
\begin{equation}\label{ineq1}
\frac{n}{n-1}\left(m-\overline{x}^{*(k)}\right)<\frac{\alpha_{(k+1)}}{1-\alpha_{(k+1)}}.
\end{equation}On the other hand, the numerator of (\ref{kdcformula}) must be non-positive for the bottom $k$ students in the course (i.e., $\overline{x}^{*(k)}_{-i}$ belongs to player $i$'s no-show reason for $i$ such that $\alpha_i\leq\alpha_{(k)}$), so that
\begin{equation}\label{ineq2}
\frac{\alpha_{(k)}}{1-\alpha_{(k)}}\leq\frac{n}{n-1}\left(m-\overline{x}^{*(k)}\right).
\end{equation}After re-arranging and simplifying the inequalities (\ref{ineq1}) and (\ref{ineq2}), we obtain the stated characterization (\ref{kdcinterval}).  In case $k=n$, we will have a nobody-cares equilibrium if and only if the zero vector belongs to the no-show region of all players;  thus, we have an $n$-don't care equilibrium if and only if (\ref{ineq2}) holds for $k=n$.  The latter condition simplifies to $\alpha_{(n)}\leq m/\left(m+1-1/n\right)$, as promised above.  Q.E.D.  
\end{proof}
\end{corollary}
Note that the middle expression in (\ref{kdcinterval}), which is bracketed by $\left[\alpha_{(k)},\alpha_{(k+1)}\right)$, is strictly increasing in the aggregate number of free exam points $n\left(m-\overline{x}^{*(k)}\right)$ that are given away by the instructor in equilibrium.  Thus, the condition (\ref{kdcinterval}) says that the total number of free points must be large enough that $k$-don't care, but small enough that the $(k+1)^{st}$ weakest student in the course does care to show up for the exam and exert strictly positive effort.

\begin{figure}[t]
\begin{center}
\includegraphics[height=200px]{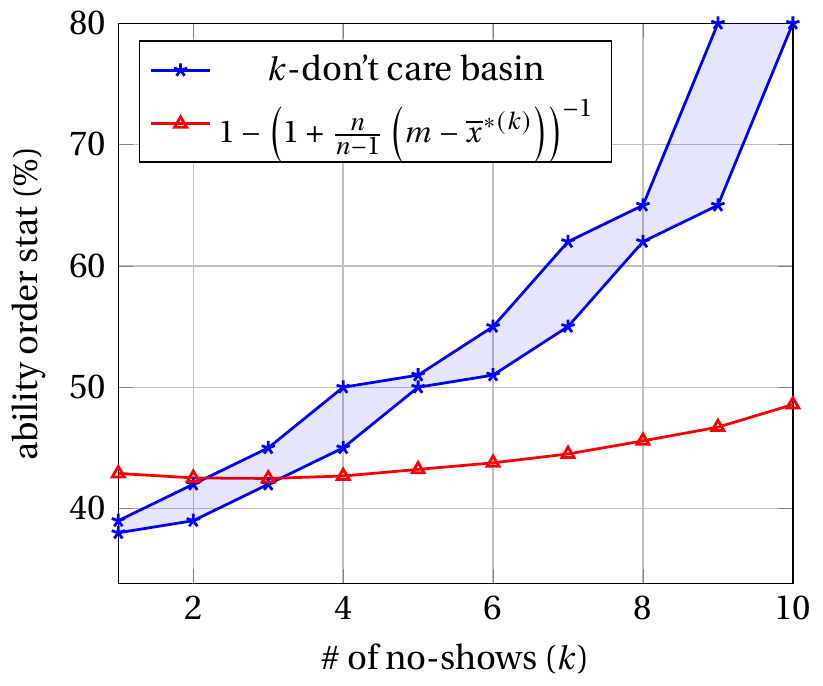}
\caption{\sc Illustration of a $3$-don't care equilibrium in a 10-student course, for the parameters $m:=85\%$ and $\alpha:=\left(.38,.39,.42,.45,.5,.51,.55,.62,.65,.8\right)$.  The number $1-\left(1+(10/9)\left(0.85-\overline{x}^{*(3)}\right)\right)^{-1}=0.4248$ lies between the third and fourth order statistics, $\alpha_{(3)}=0.42$ and $\alpha_{(4)}=0.45$.  Thus, we have a Nash equilibrium whereby the bottom 3 students are no-shows.}
\end{center}
\end{figure}

We proceed to give a complete structure theorem for the lattice of Nash equilibria (cf. with \cite{vives,berkeley}).  In  general, for a supermodular game, when we take the coordinate-wise maximum (join) or minimum (meet) of two Nash equilibria, the resulting action profile is also a Nash equilibrium.  In our particular application, we can say much more:  the Nash equilibria are totally ordered with respect to effort and with respect to Pareto preference;  the $k$-don't care equilibria form a chain in each student's commodity space, i.e., when the equilibrium number of no-shows increases, all students' grades will increase and all students' efforts will decrease.  In passing from a no-curve equilibrium to a curved interior ($0$-don't care) equilibrium (assuming that they both exist in the same model), one or more students can wind up with lower grades, but all students are guaranteed to wind up with higher utility.

\begin{theorem}[Structure of the Fixpoint Lattice]\label{comparable}
For a given model $\theta\in\Theta$, the lattice of Nash equilibria forms a chain, or linearly ordered set, with respect to vector comparisons ($\leq$) in $\mathbb{R}^n$.  That is, given two Nash equilibria $x^*$ and $y^*$, one of the two will have lower effort (and higher leisure) across the board:  either $x^*_i\leq y^*_i$ for all $i$ or else $x^*_i\geq y^*_i$ for all $i$.  
\par
The Nash equilibria of the game are totally ordered with respect to Pareto dominance;  given two equilibria $\left(x^*,y^*\right)$, all agents are better off in the low-effort equilibrium than they are in the high-effort equilibrium.  
\par
The set of $k$-don't care equilibria generates a linearly ordered set of resource allocations $\left(G_i^{*(k)},L_i^{*(k)}\right)$ that travels in a north-easterly direction in each student's commodity space $\mathbb{R}_+^2$:  given a $k$-don't care equilibrium and an $l$-don't care equilibrium, with $l>k$, all students will have higher grades and more leisure time\footnote{This is stronger than mere Pareto dominance, since all students are receiving more of both goods.  Say, when passing from the no-curve equilibrium to the $0$-don't care equilibrium (when they both exist), we are no longer guaranteed to travel north-east in each student's commodity space $\mathbb{R}^2_+$;  some or all students may wind up travelling north-west, but landing on a higher indifference curve.} under the $l$-don't care equilibrium than they do under the $k$-don't care equilibrium.
\begin{proof}
Let us start by comparing the effort levels in a $k$-don't care equilibrium with those in an $l$-don't care equilibrium, with $l>k$.  According to the characterization (\ref{kdcinterval}), a certain strictly decreasing function  
\begin{equation}
\delta(z):=1-\left(1+\frac{n}{n-1}\left(m-z\right)\right)^{-1}
\end{equation}of the $k$-don't care mean must lie in the interval $\left[\alpha_{(k)},\alpha_{(k+1)}\right)$;  and $\delta(\bullet)$ applied to the $l$-don't care mean must lie in the interval $\left[\alpha_{(l)},\alpha_{(l+1)}\right)$.  Thus, we must have $\overline{x}^{*(l)}<\overline{x}^{*(k)}$, so that the curve will be more generous in an equilibrium with $l$ no-shows.  Note well that this logic remains correct if $l=n$ (whence $\delta\left(\overline{x}^{*(n)}\right)\geq\alpha_{(n)}$) or if $k=0$ (in which case $\delta\left(\overline{x}^{*(0)}\right)<\alpha_{(1)}$, by (\ref{0dc})).
\par
Now, the (zero) effort levels of the bottom $k$ students are obviously unchanged by the transition $l\hookrightarrow k$;  the bottom $(k+1)^{st}$ through $l^{th}$ students will decrease their efforts to zero.  The top $n-l$ students in the course, who are still giving positive effort, will see those efforts (strictly) decrease, on account of the equilibrium formula (\ref{kdcformula}), which is affinely strictly increasing in the don't-care mean.  This proves that the set of Nash equilibria is a chain in $[0,1]^n$. 
\par
Next, let us consider student $i$'s welfare as we move from a $k$-don't care to an $l$-don't care equilibrium.  The curve is more generous, and he has more leisure time, but his effort has undergone a negative change $\Delta x_i^*$ across the two equilibrium points.  Looking at the formula (\ref{kdcformula}), and bearing in mind that we must have $\Delta x_i^*\geq -x_i^*$, student $i$'s effort change consists in
\begin{equation}
\Delta x_i^*=\max\left(\frac{n(1-\alpha_i)}{n-\alpha_i}\Delta\overline{x}^*,-x_i^*\right),
\end{equation}where $\Delta\overline{x}^*:=\overline{x}^{*(l)}-\overline{x}^{*(k)}<0$ is the change in the don't care mean.  If student $i$'s grade change is denoted by $\Delta G_i^*$, then we have
\begin{equation}
\Delta G_i^*=\Delta x_i^*-\Delta\overline{x}^*\geq\left(\frac{n(1-\alpha_i)}{n-\alpha_i}-1\right)\Delta\overline{x}^*=-\frac{(n-1)\alpha_i}{n-\alpha_i}\Delta\overline{x}^*>0,
\end{equation}so that all student's grades will strictly increase in the transition from a $k$-don't care to an $l$-don't care equilibrium point.
\par
Finally, in order to demonstrate that player $i$ is better off under the $0$-don't care equilibrium profile $x^{*(0)}$ than he is under the no-curve action profile $\left(\alpha_1,...,\alpha_n\right)$, we must show that\footnote{Recall that the equation on the left-hand side of (\ref{mustshow}) was the very definition of the function $\phi_i(\bullet)$ above.  The jump $J_i$ in student $i$'s reaction correspondence occurs precisely where he is indifferent between the two types of outcomes.}
\begin{equation}\label{mustshow}
\phi_i\left(\overline{x}^{*(0)}_{-i}\right)=\frac{U_i\left(x^{*(0)}\right)}{\alpha_i^{\alpha_1}(1-\alpha_i)^{1-\alpha_i}}-1\geq0,
\end{equation}where $\overline{x}^{*(0)}_{-i}$ is his opponents' sample mean in the curved interior equilibrium.  Since $\phi_i(\bullet)$ is strictly decreasing, and $\phi_i^{-1}(0)=J_i$, the inequality (\ref{mustshow}) is equivalent to $\overline{x}^{*(0)}_{-i}\leq J_i$, which, reading off from the formula for student $i$'s best reply correspondence, is true of any curved Nash equilibrium.  Q.E.D.
\end{proof}

\end{theorem}

\section{Concluding Remarks.}
This paper gave an utter and complete elaboration and solution of a supermodular game that all of my students have been playing for the past eight semesters.  In order to protect his teaching evaluations, a university professor implements an exam curve with a target mean $m$, as follows.  If the class average $\overline{x}$ is less than $m$, then the instructor will give everyone $m-\overline{x}$ free points in order to bring the mean up to an acceptable level;  if the class average is at least $m$, then no additional points are given.  Under this curving scheme, it becomes possible for one or more students to receive a curved grade that exceeds $100\%$;  in such cases there is no truncation of the grade, i.e., a $110\%$ net of the curve does not get cut down to $100\%$. 
\par
There are $n$ students in the course, who all have Cobb-Douglas preferences over a 2-dimensional commodity space that consists of grades and leisure time, or non-effort.  The elasticity $\alpha_i$ of each student $i$'s utility with respect to his grade is regarded as the ability, or quality, parameter for student $i$.  Put more delicately, students with lower values of $\alpha_i$ have a greater relative preference for leisure, and they feel the pain of studying more acutely.
\par
In the absence of an exam curve, each student's effort $x_i$ will correspond exactly to his or her grade, i.e., if you give it $90\%$, then your uncurved grade will be $90\%$, and you will have a $10\%$ allocation of laziness, or leisure time.  As a consequence of this resource constraint (that a player's uncurved grade and his leisure time must sum to $100\%$), each player's ability parameter is precisely the grade that he or she would get without the benefit of a curve.  For instance, a pupil with $\alpha_i=75\%$ is a C student, and someone with $\alpha_i=85\%$ is a B student.
\par
The resulting $n$-person interaction, whose outcomes the author continues to observe in his actual life, is a fascinating coordination game with negative spillovers, in which the various pupils' effort levels are strategic complements.  That is, if my classmates (read:  opponents) study harder, then this decreases the number of free points that I will receive from the curve, lowering my utility.  This lowering of my grade is painful, and increases the marginal utility of my effort.  Thus, if I believe that my opponents will exert high effort (say, because they have a high relative preference for a good grade), then I must rationally join in and make the necessary sacrifices.  But the coordination and complementarity also works in the reverse direction:  if I believe that my opponents are not going to try very hard, then the padding of the curve reduces my marginal returns to effort, inducing me to enjoy some additional leisure.  
\par
Speaking of marginal utility, the players' various payoffs are non-differentiable at effort profiles whose sample mean is exactly equal to the professor's target mean;  this generates a kink, or sharp corner, in the graph of each student's payoff.  In a game with smooth payoffs, strategic complementarity is characterized via the non-negativity of the cross partials, i.e., each student's marginal utility of effort must be increasing in the actions of all his opponents;  and such is the case here, away from the kinks.
\par
Thus, we resorted to Topkis' general theory of supermodularity (\cite{topkispaper,topkisbook}), wherein the notion of complementarity and increasing returns is expressed as a certain property of the forward differences of player $i$'s payoff.  Thus, given a hypothetical discrete amount $\Delta x_i$ of extra effort by kid $i$, the corresponding discrete utility gain $\Delta U_i$ is an increasing function of the efforts of all the non-$i$ players, so that the payoffs exhibit ``increasing differences,'' \'{a} la Topkis.  Although the students' payoffs are all continuous functions, the mechanics of our game duly prevent player $i$'s utility from being quasi-concave in his or her own strategy.  This generates an interesting non-convexity in the reaction correspondence, which we located and pinpointed in earnest.  To be specific, we discovered the ``indifference hyperplane,'' at which there is a jump, or bifurcation, in student $i$'s best reply;  his reaction correspondence is not single-valued over this piece of the domain (viz., it is a 2-point set).  
\par
This lack of convex-valuedness makes the usual machinery of the Kakutani fixed point theorem inapplicable;  instead, the general existence of equilibria is established by applying the Knaster-Tarski fixed point theorem to the extremal best response correspondences.  The supermodularity of the game guarantees the monotonicity of the best replies, so that the fixpoint set (of pure Nash equilibria) is a complete lattice.
\par
We took great pains to derive the exact expression for each student's best response correspondence, which enabled us to give complete formulas for all equilibria that could obtain in all possible situations.  The first step involved decomposing the domain of each player's reaction correspondence into a triplet of convex polytopes, which we called, respectively, the ``no-curve region,'' the ``curve region,'' and the ``make-or-break region.''  In the no-curve region, the efforts of the non-$i$ players are high enough to break the curve, regardless of $i$'s behavior, and so his best response is to play $\alpha_i$.  Similarly, in the curve region, the curve is made regardless of player $i$'s effort level, and his best reply is affinely increasing in the opponents' sample mean, $\overline{x}_{-i}$.  Within the curve region, if student $i$'s ability is sufficiently low, there will be a ``no-show'' subregion in which his best response is to put zero effort, i.e., he doesn't even care to show up to the exam.  On the other hand, if pupil $i$'s ability is sufficiently high, then all the pure strategies that involve less effort than his best reply to the zero vector are strictly dominated.  At the opposite extreme, all of student $i$'s pure strategies above $\alpha_i$ are strictly dominated by $\alpha_i$ itself.
\par
In the most complicated piece of the trifecta, student $i$ can make or break the curve, depending on the level of effort that he chooses.  If the opposing sample mean lies in kid $i$'s make-or-break region, then there will be a point $\hat{x}_i$ in $i$'s action set at which the class average is exactly equal to $m$.  Thus, the curve will be made or broken according as to whether $i$'s action is less than or greater than $\hat{x}_i$.  
\par
After so many Lemmas, and the manipulation of a couple or three inequalities, we found that there is a unique value $J_i$ of the opponents' sample mean (which we characterized as the sole root of a certain logarithmic equation $\phi_i(z)=0$) that causes $i$'s reaction correspondence to become double-valued.  That is, against opposing action profiles $x_{-i}$ that lie on the hyperplane $\overline{x}_{-i}=J_i$, student $i$ will have a distant pair of best responses, one high-effort ($=\alpha_i$) and one low-effort.  In this happenstance, player $i$ will be exactly indifferent between breaking the curve with high effort or making the curve with low effort.  If $x_{-i}$ is in $i$'s make-or-break region, but it lies below the indifference plane, then student $i$'s best response is to make the curve; when the opposing strategy profile lies above the indifference plane, then kid $i$'s unique reaction is to break the curve.
\par
The make-or-break region comprises a narrow strip (it is a convex polytope) within the domain of each player's reaction correspondence.  In the limit as the class size $n$ becomes large, the volumes of the make-or-break regions and the lengths of the corresponding jumps on the players' strategy axes will all converge to zero, thereby erasing our non-convexities as $n\to\infty$.  
\par
At the other end of the spectrum, smaller class sizes will tend magnify the strategic effects of the ``make-or-break'' curve mechanic.  For instance, in the 2-person game, every $2\%$ that you lose on the exam increases your opponent's score (and yours) by $1\%$;  similarly, every additional $2\%$ that you score on the exam will create a significant negative spillover for your classmate.  This effect gets dampened away in larger sections, i.e., in a 50-student course, every additional $50\%$ that you score imposes a $1\%$ externality on your peers.  In a large course, then, the main strategic consideration is not whether to make or break the curve, but rather, how to properly tune and calibrate one's effort level to the overall ability of the student population.    
\par
On the strength of our exact formulas for the best reply surfaces, we gave algebraic expressions for all points of rational play and derived the precise structure of the lattice of pure Nash equilibria.  We proved that there are $n+2$ possible types of equilibria in the model, which are totally ordered with respect to vector comparisons in $n$-dimensional Euclidean space.  The greatest possible type of Nash equilibrium (when it exists in a given model, $\theta$) is the no-curve (or ``try-hard'') outcome, whereby all students try as hard as they normally would in an uncurved course.  Next, we have the curved interior equilibrium, which we refer to as ``$0$-don't care,'' since everybody shows up to the exam. 
\par
Finally, we have the ``$k$-don't care'' equilibria, for $k=1,2,...,n$, wherein the bottom $k$ students (who have the $k$ lowest ability parameters) are no-shows.  Accordingly, we gave a complete pre-image and decomposition of the parameter space $\Theta$, resplendent with the characteristic conditions under which each type of equilibrium will or will not exist as a fixpoint of the reaction correspondence.  Although the number of equilibria in any model $\theta$ is guaranteed to be finite (between 1 and $n+2$), there is no generic expectation of oddness (as we saw in a beautiful figures and example above), on account of the infinite strategy sets and the non-convex best responses.
\par
In our comparative static analysis of the model, we found that the correct way to partially order the parameter space $\left(\Theta,\le\right)$ is with respect to the ``hardness,'' or difficulty, of the model $\theta\in\Theta$, which is considered to be increasing in the ability vector and decreasing in the instructor's target mean.  We showed that each person's payoff is properly indexed with respect to hardness, in the sense that all individual utilities have increasing differences with respect to $\theta$ under the hardness order.  Such indexation has far-reaching theoretical consequences, both for general supermodular games and for the concrete details that are specific to our university model.  Thus, the extremal best responses (which, for a given player, only differ from each other over the indifference plane) and the greatest and least pure Nash equilibria are all increasing in the hardness of the course.  In terms of learning to play the game, iterating the greatest and least best responses (seeded by the greatest and least strategy profiles, respectively) is guaranteed to converge to the greatest and least pure Nash equilibria, which, respectively, are also the greatest and least profiles of rationalizable strategies.
\par
Given two Nash equilibria, one of which will always have across-the-board less effort than the other, the low-effort equilibrium will always Pareto dominate the high-effort equilibrium.  In our university setting, we discovered that the dominance is even stronger, in the following sense.  Given a $k$-don't care equilibrium and an $l$-don't care equilibrium, with $l>k$, all students' grades will be higher (and their efforts will be lower) in the equilibrium that has more no-shows.  Hence, the set of equilibrium allocations for any given student will form a chain that moves to the north-east in his or her grade-leisure plane.  The sole exception to this north-easterly movement occurs when we pass from a no-curve equilibrium to a $0$-don't care equilibrium.  As we showed in a 2-person example above, in the aftermath of this unique equilibrium transition, some or all of the students can wind up with lower grades, so that their consumption bundles undergo north-westerly displacements in the grade-leisure plane.  However, the extra leisure time is more than compensatory for the lost points, and all students are guaranteed to land on a higher indifference curve.  Be that as it may, as the class size becomes large, even this one mode of north-westerly travel becomes impossible.
\par
The most analytically tractable type of equilibrium, and the one preferred by the author, is the curved interior equilibrium without any no-shows.  In this connection, the formulas for rational play become especially beautiful and simple.  We found that each player's behavior in the $0$-don't care equilibrium is driven by the following sufficient statistic (``$\hat{\alpha}_n$'') for the class ability vector:  we take $n$ minus the harmonic mean of the numbers $\left(n-\alpha_1,...,n-\alpha_n\right)$.  Letting $n\to\infty$ in the curved interior equilibrium, where $\hat{\alpha}_\infty$ is the sufficient statistic for the student population, we obtained the following beautiful result:  the grades and leisure allocations of all students get multiplied by the same factor $m/\hat{\alpha}_\infty$ (think 1.14) relative to the no-curve outcome.  Thus, we have a simple dilation of each student's original $\left(\text{grade},\text{leisure}\right)$ bundle, which was $\left(\alpha_i,1-\alpha_i\right)$.  
\par
The disincentive of the exam curve, (which was necessary in order to protect the professor's teaching evaluations) will therefore exacerbate (say, by a factor of $1.14$) the grade inequality\footnote{and the leisure inequality.} that was already present in the class ability vector.  That is, since all students increase their leisure hours by a uniform percentage, the top students in the course (who are starting from a low leisure base) will lose a comparatively small number of study hours;  say, a low-performing student who increases leisure by yet another $14\%$ will end up losing more study time than a high-ability student who does the same.
\par
In the opinion of this author, our game-theoretic exam curve furnishes the perfect excuse for upping the quality and enjoyability of the source material, for the mutual benefit of all the $n+1-k$ people who care.  What a great way to learn Economics.

\textit{Marathon, Greece, Summer 2021}.

\subsection*{Disclosures.}
This paper is solely the work of the (sole) author;  it was internally funded, and He has no conflicts of interest to declare.
\printbibliography
\end{document}